\documentclass[a4paper,11pt]{article}
\usepackage{etoolbox}
\usepackage{stackengine}
\newbool{conf}
\setbool{conf}{false}

\usepackage[utf8]{inputenc}
\usepackage{etoolbox}
\usepackage{amsmath, amsthm}
\usepackage{amsfonts, amssymb}
\usepackage{bm}
\usepackage{ifthen}
\usepackage{letltxmacro}
\usepackage{nicefrac}
\usepackage{xstring}
\usepackage{pgf}
\usepackage{tikz}
\usetikzlibrary{arrows,automata}
\usetikzlibrary{er,positioning,bayesnet}
\usepackage[linesnumbered,ruled]{algorithm2e}
\renewcommand{\mathbf}[1]{\bm{#1}}
\usepackage{mathrsfs}
\usepackage{graphicx}
\usepackage{framed}

\usepackage{todo}

\usepackage{booktabs,siunitx}
\usepackage{caption}
\usepackage{subcaption}
\usepackage{wrapfig}
\usepackage{cleveref}

\newtheorem{theorem}{Theorem}[section]
\newtheorem{observation}[theorem]{Observation}

\newtheorem{lemma}[theorem]{Lemma}

\newtheorem{proposition}[theorem]{Proposition}

\newtheorem{definition}[theorem]{Definition}

\newtheorem{problem}[theorem]{Problem}
\crefname{theorem}{Theorem}{Theorems}
\crefname{observation}{Observation}{Observations}
\crefname{claim}{Claim}{Claims}
\crefname{algorithm}{Algorithm}{Algorithms}
\crefname{condition}{Condition}{Conditions}
\crefname{example}{Example}{Examples}
\crefname{fact}{Fact}{Facts}
\crefname{lemma}{Lemma}{Lemmas}
\crefname{corollary}{Corollary}{Corollaries}
\crefname{consequence}{Consequence}{Consequences}
\crefname{definition}{Definition}{Definitions}
\crefname{remark}{Remark}{Remarks}
\crefname{problem}{Problem}{Problems}
\crefname{proposition}{Proposition}{Propositions}
\crefname{section}{Section}{Sections}

\crefname{appendix}{Supplementary Material Section}{Supplementary Material Sections}

\newcommand{\abs}[1]{\ensuremath{\left|#1\right|}}

\newcommand{\inb}[1]{\left\{#1\right\}}
\newcommand{\inp}[1]{\left(#1\right)}

\makeatletter
\newcommand*{\defeq}{\mathrel{\rlap{
                     \raisebox{0.3ex}{$\m@th\cdot$}}
                     \raisebox{-0.3ex}{$\m@th\cdot$}}
                    =}
\newcommand*{\eqdef}{=
  \mathrel{\rlap{
      \raisebox{0.3ex}{$\m@th\cdot$}}
    \raisebox{-0.3ex}{$\m@th\cdot$}}
}
\makeatother

\renewcommand{\emptyset}[0]{\varnothing}

\newcommand{\poly}[1]{\ensuremath{\mathop{\mathrm{poly}}\inp{#1}}}

\newcommand{\undir}[2]{\ensuremath{{#1} - {#2}}}

\usetikzlibrary{arrows}

\usepackage{fullpage}
\usepackage{appendix}
\usepackage{setspace}
\usetikzlibrary{er,positioning,bayesnet}
\usepackage{multicol}
\usepackage{verbatim}

\usepackage{enumerate}

\definecolor{blue}{HTML}{1F77B4}
\definecolor{orange}{HTML}{FF7F0E}
\definecolor{green}{HTML}{2CA02C}

\usepackage{marginnote}
\usepackage{todo}

\usepackage[american]{babel}
\usepackage{natbib} 
\usepackage{mathtools}
\usepackage{booktabs} 
\usepackage{tikz} 
\footnotesize

\begin{document}
\title{Counting Markov Equivalent Directed Acyclic Graphs Consistent with Background Knowledge\footnote{Accepted to the Proceedings of the 39th Conference on Uncertainty in Artificial Intelligence (UAI-2023)}}

\author{}
 \author{Vidya Sagar Sharma\thanks{Tata Institute of Fundamental Research, Mumbai. Email: \texttt{vidya.sagar@tifr.res.in}.} 
}

\date{}
\maketitle              

\thispagestyle{empty}
\maketitle

\begin{abstract}
  We study the problem of counting the number of directed acyclic graphs in a
  Markov equivalence class (MEC) that are consistent with \emph{background
    knowledge} specified in the form of the directions of some additional edges
  in the MEC.  A polynomial-time algorithm for the special case of the problem, when no background knowledge constraints are specified, was given by Wienöbst,
  Bannach, and Liśkiewicz (AAAI 2021), who also showed that the general case is
  NP-hard (in fact, \#P-hard).  In this paper, we show that the problem is
  nevertheless tractable in an interesting class of instances, by establishing
  that it is ``fixed-parameter tractable'': we give an algorithm that runs in
  time $O(k! k^2 n^4)$, where $n$ is the number of nodes in the MEC and $k$ is the
  maximum number of nodes in any maximal clique of the MEC that participate in
  the specified background knowledge constraints.  In particular, our algorithm
  runs in polynomial time in the well-studied special case of MECs of bounded
  tree-width or bounded maximum clique size.
\end{abstract}

 \section{Introduction}
\label{sec:introduction}
A graphical model is a combinatorial tool for expressing dependencies between random variables. Both directed and undirected versions have been used in the literature for modeling different kinds of dependency structures.  We study graphical models represented by directed acyclic graphs (DAGs), which represent conditional independence relations and causal influences between random variables by directed edges~\citep{Pearl2009}. Such graphical models have been used extensively for modeling causal relationships across several fields, e.g., material science~\citep{ren_embedding_2020}, game theory~\citep{kearns2013graphical}, and biology~\citep{friedman2004inferring,finegold2011robust}.

It is well known that given access only to observational data, the causal DAG underlying a system can only be determined up to its ``Markov equivalence class'' (MEC)~\citep{verma1991equivalence,meek2013causal,chickering_causal_1995}.  Two DAGs are said to be in the same MEC if they model exactly the same set of conditional dependence relations between the underlying random variables. Distinguishing between two DAGs in the same MEC requires the use of interventional data~\citep{HB12}.

Finding the size of an MEC, therefore, becomes a question of key interest. In particular, the size of an MEC quantifies the uncertainty of the causal model given only observational data.  The problem, along with a proposed algorithm, was already mentioned by Meek~\citep[Section 4.1]{meek2013causal}, and has since then been the focus of a long line of work~\citep{madigan1996bayesian,he2015counting,he2016formulas,bernstein2017sampling,ghassami2019counting,talvitie2019counting,ganian2020efficient}, which culminated in a polynomial time algorithm for the problem by \citet{wienobst2020polynomial}.

\paragraph{Counting with background knowledge constraints} In applications, more
information about the directions of edges in the underlying DAG than that
encoded in the MEC may be available, for example, due to access to
domain-specific knowledge.  \citet{meek2013causal} referred to this as
\emph{background knowledge} and modeled it as a specification of the directions
of some of the edges of the underlying DAG. Thus, instead of finding the size of
the whole MEC, one becomes interested in counting those DAGs in the MEC that are
consistent with this specified background knowledge.  An algorithm for this
problem can also be used as an indicator of the efficacy of a particular
intervention in pinning down a DAG within an MEC, by measuring the ratio between
the size of the MEC and the number of those DAGs in the MEC that are consistent
with the extra background information yielded by the intervention.
\citet{wienobst2020polynomial} showed, however, that in general, this problem is
\#P-complete (i.e., as hard as counting satisfying assignments to a Boolean
formula). Our goal in this paper is to circumvent this hardness result when the
specified background knowledge has some special structure.

\subsection{Our Contributions}
We now formalize the above problem.  Our input is an MEC on $n$ nodes, and we
are given also the directions of another $s$ edges that are not directed in the
MEC: we refer to the set of these edges as the \emph{background knowledge},
denoted $\mathcal{K}$. Our goal is to count the number of DAGs in the input MEC
that are consistent with the background knowledge $\mathcal{K}$. The above
quoted result of \cite{wienobst2020polynomial} implies that (under the standard
P $\neq$ NP assumption) there \emph{cannot} be an algorithm for solving this
problem whose run-time is polynomial in both $n$ and $s$.

\paragraph{Main result}
The main conceptual contribution of this paper is to
define the following parameter which lets us identify important special
instances of the problem where we can circumvent the above hardness result.
Given the set $\mathcal{K}$ of background knowledge edges, we define the
\emph{max-clique-knowledge} $k$ of $\mathcal{K}$ to be the maximum number of
vertices in any clique in the input MEC that are part of a background knowledge
edge that lies completely inside that clique.  In particular, $k$ can be at most
twice $s$, but it can also be much smaller.  Our main result
(\cref{thm:main-result}) is an algorithm that counts the number of DAGs in the
MEC that are consistent with $\mathcal{K}$, and runs in time
$O(k! \cdot k^2 \cdot n^4)$.

\paragraph{Discussion and evaluation} In particular, the runtime of our
algorithm is polynomially bounded when the parameter $k$ above is bounded above
by a constant, \emph{even if} the actual size $s$ of the background knowledge is
very large.  For example, since $k$ is bounded above by the size of the largest
clique in the input MEC, it follows that our algorithm runs in polynomial time
in the well-studied special case (see, e.g., \citep{talvitie2019counting}) when
the input MEC has constant tree-width (and hence constant maximum-clique
size). We provide an empirical exploration of the run-time of our algorithm,
including of the phenomenon that it depends on $\mathcal{K}$ only through $k$
and not through $s$, in \cref{sec:experimental_evaluation}.

\subsection{Related Work}
It is well known that an MEC can be represented as a partially directed graph,
known as an \emph{essential graph}, with special graph theoretic
properties~\citep{verma1991equivalence,meek2013causal,chickering_causal_1995,andersson1997characterization}. Initial
approaches to the problem of computing the size of an MEC
by~\citet{meek2013causal} and \cite{madigan1996bayesian} built upon ideas
underlying this characterization. More recently, \cite{he2015counting} evaluated
a heuristic based on the partition of essential graphs into chordal
components. \citet{ghassami2019counting} gave an algorithm that runs in
polynomial time for constant degree graphs, but where the degree of this
polynomial \emph{grows} with the maximum degree of the graph.

Our algorithm can be seen as an example of a \emph{fixed parameter tractable}
(FPT) algorithm in the well-studied framework of parameterized complexity
theory~\citep{cygan_parameterized_2015}.  Parameterized complexity offers an
approach to attack computationally hard problems (e.g. those that are NP-hard or
\#P-hard) by separating the complexity of solving the problem into two pieces --
a part that depends purely on the size of the input, and a part that depends
only on a well-chosen ``parameter'' $\rho$ of the problem.  An FPT algorithm for
a computationally hard problem has a runtime that is bounded above by
$f(\rho)\poly{n}$, where the degree of the polynomial does \emph{not} depend
upon the parameter $\rho$, but where the function $f$ (which does not depend
upon the input size $n$) may potentially be exponentially growing. In our setting, the underlying parameter is the max-clique-knowledge of the background
knowledge $\mathcal{K}$.

Thus, the algorithm of \cite{ghassami2019counting} cited above is \emph{not} an
FPT algorithm.  However, \citet{talvitie2019counting} improved upon it by giving
an FPT algorithm for computing the size of an MEC
(without background knowledge): for an undirected essential graph on $n$ nodes
whose maximum clique is of size $c$, their algorithm runs in time
$\mathcal{O}(c!2^c c^2n)$.  Finally, \citet{wienobst2020polynomial} presented the first polynomial time
algorithm for computing the size of any MEC (again, without background
knowledge).  As discussed above, they also showed that counting Markov
equivalent DAGs consistent with specified background knowledge is, in general,
\#P-complete.

We are not aware of any progress towards circumventing this hardness result by
imposing specific properties on the specified background knowledge, and to the
best of our knowledge, this paper is the first to give a fixed parameter
tractable algorithm for the problem. Our algorithm is motivated by the
techniques developed by \citet{wienobst2020polynomial}.

 \section{Preliminaries}
\label{sec:preliminaries}
We mostly follow the terminology and notation used by
\cite{andersson1997characterization} and \cite{wienobst2020polynomial} for
notions such as \textbf{\emph{graph unions}, \emph{chain graphs},
  \emph{directed} and \emph{undirected} graphs, \emph{skeletons}},
\textbf{v-\emph{structures}}, \textbf{\emph{cliques, separators, chordal graphs, undirected connected chordal
    graphs}} (which we will typically denote using the abbreviation
  \textbf{UCCG}), and \textbf{\emph{clique trees}}.  For completeness, we
  provide detailed definitions in the
Supplementary Material.

\textbf{Notations. } A graph $G$ is a pair $(V, E)$, where $V$ is said to be the set of vertices of $G$, and $E\subseteq V\times V$ is said to be the set of edges of $G$. For $u,v \in V$, if $(u,v),(v,u) \in E$ then we say there is an undirected edge between $u$ and $v$, denoted as $u-v$. For $u,v \in V$, if $(u,v) \in E$, and $(v,u)\notin E$ then we say there is a directed edge from $u$ to $v$, denoted as $u\rightarrow v$. For a graph $G$, we denote $V_G$ as the subset of vertices of $G$, and $E_G$ as the set of vertices of $G$. A clique is a set of pairwise adjacent vertices. We denote the set of all maximal cliques of $G$ by $\Pi(G)$.
For a set $X$, we denote by $\# X$ (and sometimes by $|X|$), the size of
$X$.

\textbf{Markov equivalence classes.} A causal DAG encodes a set of conditional independence relations between random variables represented by its vertices. Two DAGs are said to belong to the same \emph{Markov equivalence class} (MEC) if both encode the same set of conditional independence relations. \citet{verma1991equivalence} showed that two DAGs are in the same MEC if, and only if, both have (i) the same skeleton and (ii) the same set of v-structures.  An MEC can be represented by the graph union of all DAGs in it. \citet{andersson1997characterization} show that a partially directed graph representing an MEC is a chain graph whose undirected connected components (i.e., undirected connected components formed after removing the directed edges) are chordal graphs. We refer to these undirected connected components as \emph{chordal components} of the MEC.  With a slight abuse of terminology, we equate the chain graph with chordal components which represents an MEC with the MEC itself, and refer to both as an ``MEC''. 

\textbf{AMO.} Given a partially directed graph $G$, an \emph{orientation} of $G$ is obtained by assigning a direction to each undirected edge of $G$.  Following \citet{wienobst2020polynomial}, we call an orientation of $G$ an \emph{acyclic moral orientation} (AMO) of $G$ if (i) it does not contain any directed cycles, and (ii) it has the same set of v-structures as $G$. For an MEC $G$, we denote the set of AMOs of $G$ by AMO($G$).

\textbf{PEO and LBFS orderings.} For an undirected graph, a linear ordering $\tau$ of its vertices is said to be a \emph{perfect elimination ordering} (PEO) of the graph if for each vertex $v$, the neighbors of $v$ that occur after $v$ form a clique. A graph is chordal if, and only if, it has a PEO \citep{fulkerson1965incidence}. \citet{rose1976algorithmic} gave a \emph{lexicographical breadth-first-search} (LBFS) algorithm to find a PEO of a chordal graph (see \cref{alg:AMO-Union-K}  below for a modified version of LBFS).  Any ordering of vertices that can be returned by the LBFS algorithm is said to be an \emph{LBFS ordering}.

\textbf{Representation of an AMO.}
\label{preliminaries:representation-of-an-AMO}
Given a linear ordering $\tau$ of the vertices of a graph $G$, an AMO of $G$ is said
to be \emph{represented} by $\tau$ if for every edge $u \rightarrow v$ in the
AMO, $u$ precedes $v$ in $\tau$.  Every LBFS ordering of a UCCG $G$ represents
a unique AMO of $G$, and every AMO of a UCCG $G$ is represented by an LBFS
ordering of $G$ (Corollary 1, and Lemma 2 of
\citet{wienobst2020polynomial}). For a maximal clique $C$ of $G$, we say that
\emph{$C$ represents an AMO $\alpha$ of $G$} if there exists an LBFS ordering that
starts with $C$ and represents $\alpha$. Similarly, for a permutation $\pi(C)$ of a
maximal clique $C$ of $G$, we say $\pi(C)$ represents $\alpha$ if
there exists an LBFS ordering that starts with $\pi(C)$ and represents $\alpha$. We denote by AMO$(G,\pi(C))$ and AMO$(G, C)$, the set of AMOs of $G$ that can be represented by $\pi(C)$ and $C$, respectively.

\textbf{Canonical representation of AMO.} For an AMO $\alpha$ of a UCCG $G$,
\cite{wienobst2020polynomial} define a unique clique that represents
$\alpha$. We denote the clique as $C_{\alpha}$, and say that $C_{\alpha}$
\emph{canonically} represents $\alpha$.  For the purposes of this paper, we only
need certain properties of the canonical representative $C_\alpha$, and these
are quoted in \cref{lem:identification-of-C-alpha}. However, for completeness,
we also give the definition of $C_{\alpha}$ in the Supplementary Material.

\begin{figure}[h!]
    \centering
    \begin{tikzpicture}
        \node[scale = 0.9](a){$a$};
        \node[scale = 0.9](b)[below left=0.5 and 0.5 of a] {$b$};
        \node[scale = 0.9](c)[below right=0.5 and 0.5 of a] {$c$};
        \node[scale = 0.9](d)[below =0.5 of c]{$d$};
        \node[scale = 0.9](e)[below left=0.5 and 0.5 of d] {$e$};
        \node[scale = 0.9](f)[below right=0.5 and 0.5 of d] {$f$};
        
        \draw[-](a)--(b);
        \draw[->](a)--(c);
        \draw[->](b)--(c);
        \draw[<-](c)--(d);
        \draw[-](d)--(e);
        \draw[-](d)--(f);
        \draw[-](e)--(f);

        \node[scale = 0.9](1)[right = 3.0 of a]{1};
        \node[scale = 0.9](2)[right=1.0 of 1]{2};
        \node[scale = 0.9](3)[below = 1.0 of 1]{3};
        \node[scale = 0.9](4)[right = 1.0 of 3]{4};
        \node[scale = 0.9](5)[below = 1.0 of 3]{5};
        \node[scale = 0.9](6)[right = 1.0 of 5]{6};
        \node[scale = 0.9](7)[below = 1.0 of 6]{7};
        \draw[-](1)--(2){};
        \draw[-](1)--(3){};
        \draw[-](1)--(4){};
        \draw[-](2)--(3){};
        \draw[-](2)--(4){};
        \draw[-](3)--(4){};
        \draw[-](3)--(5){};
        \draw[-](3)--(6){};
        \draw[-](4)--(5){};
        \draw[-](4)--(6){};
        \draw[-](5)--(6){};
        \draw[-](5)--(7){};
        \draw[-](6)--(7){};
        \node[scale = 0.9](mec1)[above =0.1 of a]{MEC 1};
        \node[scale = 0.9](mec2)[above right =0.1 and -0.1 of 1]{MEC 2};

        \node[scale =0.9](table-11)[below left=0.5 and 1.5 of 7]{Background Knowledge};
        \node[scale = 0.9](table-12)[right =0.0 of table-11]{Max-clique Knowledge};
        \node[scale = 0.9](eg-11)[below left= 0.1 and 3.5 of table-12]{Example 1};
        \node[scale = 0.9](bkeg-11)[right=0.0 of eg-11]{$\{a\rightarrow b$, $e\rightarrow d$, $f\rightarrow d\}$};
        \node[scale = 0.9](ckeg-11)[right=5.0 of eg-11]{3};
        \node[scale = 0.9](eg-12)[below=0.1 of eg-11]{Example 2};
        \node[scale = 0.9](bkeg-12)[right=0.0 of eg-12]{$\{a\rightarrow b$, $e\rightarrow d\}$};
        \node[scale = 0.9](ckeg-12)[right=5.0 of eg-12]{2};

\node[scale = 0.9](eg-21)[below = 0.1 of eg-12]{Example 3};
        \node[scale = 0.9](bkeg-21)[right=0.0 of eg-21]{$\{1\rightarrow 2$, $3\rightarrow 6\}$};
        \node[scale = 0.9](ckeg-21)[right=5.0 of eg-21]{2};

        \node[scale = 0.9](eg-22)[below = 0.1 of eg-21]{Example 4};
        \node[scale = 0.9](bkeg-22)[right=0.0 of eg-22]{$\{1\rightarrow 2$, $2\rightarrow 3$, $1\rightarrow 3$,};
        \node[scale = 0.9](bkeg-221)[below=0.1 of bkeg-22]{$3\rightarrow 6$, $4\rightarrow 6$, $6\rightarrow 7\}$};
        \node[scale = 0.9](ckeg-22)[right=5.0 of eg-22]{3};

    \end{tikzpicture}
    \caption{Background Knowledge and Max-clique Knowledge:  In MEC 1, there are 2 maximal cliques $\{a,b\}$ and $\{d,e,f\}$. In MEC 2, there are 3 maximal cliques $\{1,2,3,4\}$, $\{3,4,5,6\}$ and $\{5,6,7\}$. Examples 1 and 2 are for MEC 1, and examples 3 and 4 are for MEC 2. In example 1, for the maximal clique $\{d,e,f\}$, $d$ and $e$ are part of an edge $e\rightarrow d$, and $f$ is part of an edge $f\rightarrow d$, i.e., all the nodes of the clique is part of an edge of the background knowledge such that both endpoints of the edge are inside the clique. From the definition of clique knowledge, clique knowledge of the clique $\{d,e,f\}$ is 3. Similarly, clique knowledge of the maximal clique $\{a,b\}$ is 2. This shows the max-clique knowledge of the background knowledge in example 1 for MEC 1 is 3. Similarly, the max-clique knowledge of the background knowledge in example 2 for MEC 1 is 2, and the max-clique knowledge of the background knowledge in examples 3 and 4 for MEC 2 are 2 and 3, respectively.}
    \label{fig:background-knowledge-andclique-knowledge}
\end{figure} 
\textbf{Background Knowledge and Clique-knowledge.}
\label{def:background-knowledge-and-clique-knowledge}
For an MEC represented by a
partially directed graph $G$, \emph{background knowledge}\footnote{ Another way
  to represent background knowledge is by a maximally partially directed acyclic
  graph (MPDAG). We do not use MPDAGs in our paper because our run time
  depends only on the number of directed background knowledge edges that
  generate the MPDAG, and not on the (possibly much larger) number of directed
  edges in the MPDAG.} is specified as a set of directed edges
$\mathcal{K} \subseteq E_G$ \citep{meek2013causal}. A graph $G$ is said to be
consistent with $\mathcal{K}$ if, for any edge
$u \rightarrow v \in \mathcal{K}$, $v \rightarrow u \notin E_G$ (i.e., either $u-v\in E_G$ or $u\rightarrow v \in G$).  For a clique
$C$ of $G$, \emph{clique-knowledge} of $\mathcal{K}$ for $C$ is defined as the
number of vertices of $C$ that are part of a background knowledge edge both of whose endpoints are in $C.$
\emph{Max-clique-knowledge} of $\mathcal{K}$ for $G$ is the maximum value of
clique-knowledge over all cliques of $G$.

We denote the set of AMOs of $G$ consistent with background knowledge $\mathcal{K}$ by
AMO($G,{\mathcal{K}}$). Further, for any clique $C$ and any permutation $\pi(C)$
of the vertices of $C$, we denote by AMO$(G,\pi(C),\mathcal{K})$ and
AMO$(G,C,\mathcal{K})$, respectively, the set of $\mathcal{K}$-consistent AMOs
of $G$ that can be represented by $\pi(C)$ and
$C$.

 \section{Main Result}
\label{sec:main-result}
\citet{andersson1997characterization} show that a DAG is a member of an MEC $G$ if, and only if, it is an AMO of $G$. Thus, counting the number of DAGs in the MEC represented by $G$ is equivalent to counting AMOs of $G$.
We start with a formal description of the algorithmic problem we address in this paper.
\begin{problem}[\textbf{Counting AMOs with Background Knowledge}] \label{prob:countingAMO}
\textbf{INPUT:} (a) An MEC $G$ (in the form of a chain graph with chordal undirected components), (b) Background  Knowledge $\mathcal{K}\subseteq E_G$.

  \noindent \textbf{OUTPUT:} Number of DAGs in the MEC $G$ that are consistent  with $\mathcal{K}$, i.e., \#AMO($G, {\mathcal{K}}$).
\end{problem}

As discussed in the introduction, a polynomial time algorithm for the special
case of this problem where $\mathcal{K}$ is empty was given by
\citet{wienobst2020polynomial}, in the culmination of a long line of work on
that case. However, as discussed in the Introduction, a hardness result proved
by \citet{wienobst2020polynomial}
implies, under the standard $P \neq NP$ assumption, that there cannot
be an algorithm for \cref{prob:countingAMO} that runs in time polynomial in both
$n$ and $\abs{\mathcal{K}}$.

We, therefore, ask: what are the other interesting cases of the problem which admit an efficient solution?  We answer this question with the following result.

\begin{theorem}[Main result]\label{thm:main-result}
  There is an algorithm for \cref{prob:countingAMO} which outputs
  $\#AMO(G, \mathcal{K})$ in time $O(k!k^2n^4)$, where $k$ is the max-clique-knowledge value of $\mathcal{K}$ for $G$, and $n$ is the number of vertices in $G$.
\end{theorem}
 
In the formalism of parameterized algorithms \citep{cygan_parameterized_2015}, 
the above result says that the problem of counting AMOs that are consistent with background knowledge is \emph{fixed parameter tractable}, with the parameter being the max-clique-knowledge of the background knowledge.
This contrasts with the \#P-hardness result of \citet{wienobst2020polynomial} for the problem. 

The starting point of our algorithm is a standard reduction to the following special case of the problem.
\begin{problem}[\textbf{Counting Background Consistent AMOs in chordal
    graphs}] \label{prob:countingAMOofUCCG} \textbf{INPUT:} (a) An undirected
  connected chordal graph (UCCG) $G$, (b) Background Knowledge
  $\mathcal{K}\subseteq E_G$.

  \noindent \textbf{OUTPUT:} Number of AMOs of $G$ that are consistent with
  $\mathcal{K}$, i.e., \#AMO($G, {\mathcal{K}}$).
\end{problem}

\begin{proposition}
  \label{lem:Counting-AMO-reduction}
  Let $G$ be an MEC, and let $\mathcal{K} \subseteq E_G$ be background knowledge consistent with $G$. Then,
  \#AMO$(G,\mathcal{K}) = \prod_{\text{H}}{\#\text{AMO}(H,\mathcal{K}[H])},$
  where the product ranges over all undirected connected chordal components $H$
  of the MEC $G$, and
  $\mathcal{K}[H] = \{u\rightarrow v: u,v\in H \text{, and } u\rightarrow v\in
  \mathcal{K} \}$ is the corresponding background knowledge for $H$.
\end{proposition}
\cref{lem:Counting-AMO-reduction} reduces \cref{prob:countingAMO} to
\cref{prob:countingAMOofUCCG}.  The proof of \cref{lem:Counting-AMO-reduction}
follows directly from standard arguments (a similar reduction to chordal
components of an MEC has been used in many previous works), and is given in the
Supplementary Material. 

The rest of the paper is devoted to providing an efficient algorithm for \cref{prob:countingAMOofUCCG}. Our strategy builds upon the recursive framework developed by \citet{wienobst2020polynomial}.  In the first step, in \cref{sec:reduction-of-the-problem}, we modify the LBFS algorithm presented by \citet{wienobst2020polynomial} to make it background aware. This algorithm is used to generate smaller instances of the problem recursively.  Further, we modify the recursive algorithm of \citet{wienobst2020polynomial}  to use this new LBFS algorithm.  While building upon prior work, both steps require new ideas to take care of the background information.  In particular, it is a careful accounting of the background knowledge $\mathcal{K}$ that requires the $k!$ factor in the runtime, where $k$ is the max-clique-knowledge of the background knowledge.
In \cref{sec:time-complexity}, we analyze the time complexity of the resulting algorithm. Finally, \cref{sec:experimental_evaluation} shows our experimental results.  Due to space constraints, proofs of many lemmas and theorems are provided in the Supplementary Material. For most of the important results, we provide the proof sketch in the main text.

 \section{The Algorithm}
\label{sec:reduction-of-the-problem}
In this section, we give an FPT algorithm \cref{alg:Count-AMO} that solves \cref{prob:countingAMOofUCCG} using a parameter ``max-clique knowledge'' (\cref{sec:preliminaries}). \cref{alg:Count-AMO} is a background aware version of Algorithm 2 of \cite{wienobst2020polynomial}, which solved the special case of \cref{prob:countingAMOofUCCG} when there is no background knowledge. Similar to their algorithm, our algorithm uses a modified LBFS algorithm, \cref{alg:AMO-Union-K}, which is a background aware version of the modified LBFS algorithm presented by them. The simple \cref{alg:valid-permutaion-of-clique}, which counts background knowledge consistent permutations of a clique, is the new ingredient required in our final algorithm presented in \cref{alg:Count-AMO}.

We start with the partitioning of the AMOs. In \cref{sec:preliminaries}, we saw that each AMO of a UCCG is canonically represented by a unique maximal clique of the UCCG. We use this for partitioning the AMOs.

\begin{lemma}
\label{lem:partition-of-set-of-bkc-AMOs}
Let $G$ be a UCCG, and $\mathcal{K}$ be a given background knowledge. Then
$\#\text{AMO}(G,\mathcal{K})$ equals 
\begin{equation}
  \sum_{C\in \Pi(G)}{|\{\alpha: \alpha \in \text{AMO}(G,\mathcal{K}) \text{ and } C=C_{\alpha}\}|}. 
\end{equation}
\end{lemma}

Here, $\{\alpha: \alpha \in \text{AMO}(G,\mathcal{K}) \text{ and } C=C_{\alpha}\}$ is the set of $\mathcal{K}$-consistent AMOs of $G$ that are \emph{canonically} represented by a maximal clique $C$ of $G$. To compute this, we first compute the union of AMOs of $G$ that are represented by $C$. 
The following definition concerns objects from the work of \cite{wienobst2020polynomial} that are relevant to construct the union graph.

\begin{definition}[{$G^{\pi(C)}, G^C, \mathcal{C}_G(\pi(C)$ and $\mathcal{C}_G(C)$}, \citet{wienobst2020polynomial}, Definition 1]
\label{def:definitions-of-wienbost-paper}
Let $G$ be a UCCG, $C$ a maximal clique of $G$, and $\pi(C)$ a permutation
of $C$.  Then, $G^C$ (respectively, $G^{\pi(C)}$) denotes the union of all the
AMOs of $G$ that can be represented by $C$ (respectively, by $\pi(C)$).
$\mathcal{C}_{G}(C)$ (respectively, $\mathcal{C}_{G}(\pi(C))$) denotes the
undirected connected components of $G^C[V_G\setminus C]$ (respectively,
$G^{\pi(C)}[V_G\setminus C]$).
\end{definition}

 The structure of $G^C$ provides us the set of directed edges in $G^C$, which helps us to check the $\mathcal{K}$-consistency of $G^C$. Since $G^C$ is the union of all the AMOs that can be represented by $C$,  a directed edge in $G^C$ is a directed edge in all the AMOs that can be represented by $C$. Thus, if any such directed edge is not $\mathcal{K}$-consistent then none of the AMOs that can be represented by $C$ can be $\mathcal{K}$-consistent. And, if all the directed edges of $G^C$ are $\mathcal{K}$-consistent then we further reduce our problem into counting $\mathcal{K}$-consistent AMOs for the undirected connected components of $G^C$ (\cref{lem:final-counting-algorithm}).

In order to implement the above discussion, the main insight required is the following definition.

\begin{definition}[$\mathcal{K}$-consistency of $G^C$]
  \label{def:bc-CGC}
  Let $G$ be a UCCG and $C$ a maximal clique in $G$. Given background
  knowledge $\mathcal{K}$ about the directions of the edges of $G$, $G^C$is said
  to be $\mathcal{K}$-consistent, if there does not exist a directed edge
  $u \rightarrow v \in G^C$ such that $v\rightarrow u \in \mathcal{K}$.
\end{definition}

As discussed above, \cref{def:bc-CGC} implies that for a maximal clique $C$ of $G$, if $G^C$ is not $\mathcal{K}$-consistent then there exists no $\mathcal{K}$-consistent AMO of $G$ that is represented by $C$. In other words, if $G^C$ is not $\mathcal{K}$-consistent then \[{|\{\alpha: \alpha \in \text{AMO}(G,\mathcal{K}) \text{ and } C=C_{\alpha}\}|}=0.\] Then, from \cref{lem:partition-of-set-of-bkc-AMOs},
\begin{equation*}
  \#\text{AMO}(G,\mathcal{K}) = \sum_{C:G^C \text{ is } \mathcal{K}\text{-consistent}}{|\{\alpha: \alpha \in \text{AMO}(G,\mathcal{K}) \text{ and } C=C_{\alpha}\}|}.
\end{equation*}
We  construct \cref{alg:AMO-Union-K} to check the $\mathcal{K}$-consistency of $G^C$, for any maximal clique $C$ of $G$.
\citet{wienobst2020polynomial} give a modified LBFS algorithm that for input a chordal graph $G$, and a maximal clique $C$ of $G$, outputs the undirected connected components (UCCs) of $\mathcal{C}_G(C)$. Their algorithm outputs the UCCs of $\mathcal{C}_G(C)$ in such a way that by knowing the  UCCs of $\mathcal{C}_G(C)$, we can construct $G^C$. We use this fact and construct an LBFS algorithm \cref{alg:AMO-Union-K}, which also checks the $\mathcal{K}$-consistency of $G^C$. \cref{alg:AMO-Union-K} is a background aware version of the LBFS algorithm given by \citet{wienobst2020polynomial}.\footnote{If \cref{alg:AMO-Union-K} is executed with $C=\mathcal{K}=\emptyset$, the algorithm performs a normal LBFS with corresponding traversal ordering $\tau$, which is the reverse of a PEO of $G$.} 
 
 \begin{algorithm}[ht]
  \caption{LBFS($G, C,\mathcal{K}$) (Background aware LBFS, based on the modified LBFS of
    \citet{wienobst2020polynomial})}
\label{alg:AMO-Union-K}
\SetAlgoLined
\SetKwInOut{KwIn}{Input}
\SetKwInOut{KwOut}{Output}
\SetKwFunction{AMO-Union}{AMO-Union}
\KwIn{A UCCG $G$, a maximal clique $C$ of $G$, and background knowledge $\mathcal{K}\subseteq E_G$.}
    \KwOut{ $(1, \mathcal{C}_G(C))$: if $G^C$ is $\mathcal{K}$-consistent,\\
    $(0, \mathcal{C}_G(C))$: otherwise.
     }
$\mathcal{S}\leftarrow$ sequence of sets initialized with $(C,V\setminus C)$\label{alg:S-init}
    
    $\tau \leftarrow$ empty list, $\mathcal L\leftarrow$ empty list, $\texttt{flag}\leftarrow 1$, $Y\leftarrow$ empty list \label{alg:flag-init}
    
    \While{$\mathcal{S}$ is non-empty\label{alg:loop}}{
    $X\leftarrow$ first non-empty set of $\mathcal{S}$\label{alg:set-X}
    
    $v\leftarrow$ arbitrary vertex from $X$\label{alg:mod1}

    {
   \If{$v$ is neither in a set in $\mathcal L$ nor in $C$\label{alg:if-1-start}}{
    Append $X$ to the end of the list $\mathcal{L}$.
    
    Append undirected connected components of $G[X]$ to the end of $Y$.\label{alg:appendY}
    
}\label{alg:if-1-end}
    }
    
    Add vertex $v$ to the end of $\tau$.
    
    \If{$u \rightarrow v \in \mathcal{K}$ for any $u$ which is neither in a set in $\mathcal{L}$ nor in $C$\label{alg:mod-if}}
    {
      $\texttt{flag}=0$;\label{alg:mod-output-empty}
    }\label{alg:mod-if-end}

Replace the set $X$ in the sequence $\mathcal{S}$ by the set $X \setminus \inb{v}$. 
    
    $N(v)\leftarrow \{ x|x\notin \tau \textup{ and } \undir{v}{x}\in E\}$\label{alg:neighbour-v}
    
    Denote the current $\mathcal{S}$ by $(S_1,\ldots ,S_k)$.
    
    Replace each $S_i$ by $S_i \cap N(v), S_i \setminus  N(v)$.\label{alg:refine}
    
    Remove all empty sets from $\mathcal{S}$.
  }\label{alg:while-end}

\KwRet $(\texttt{flag}, Y)$ \label{alg:return}
\end{algorithm}
 
 We now describe our background-aware version of the modified LBFS algorithm:
 see \cref{alg:AMO-Union-K}.  We do not change any line from the LBFS algorithm
 of \citet{wienobst2020polynomial}. The modifications we do in their LBFS
 algorithm are (a) introduction of ``flag'', at \cref{alg:flag-init}, which is
 used to check the $\mathcal{K}$-consistency of $G^C$, (b) lines
 \ref{alg:mod-if}-\ref{alg:mod-if-end}, which is used to update the value of
 ``flag'', and (c) we also output the value of ``flag'' with $\mathcal{C}_G(C)$.
 The correctness of this modification (stated formally in
 \cref{lem:output-of-Alg-1}) is based on the following observation which in turn
 uses ideas implicit in the work of \citet{wienobst2020polynomial}.
 
\begin{observation}[Implicit in the work of \citet{wienobst2020polynomial}]
\label{lem:direction of-edges-of-GC-in-alg}
Let $G$ be a UCCG, $\mathcal{K}$ be the known background knowledge about $G$, and $C$ be a maximal clique of $G$. For input $G,C,$ and $\mathcal{K}$, suppose that at some iteration of \cref{alg:AMO-Union-K}, $\mathcal{L}=\{X_1,X_2,\ldots,X_l\}$.  Then,
\begin{enumerate}
    \item 
    \label{item-1-of-direction of-edges-of-G^C-in-alg}
    For any $u\in C$, and $v\notin C$, if $(u,v)\in E_G$ then $u\rightarrow v$ is a directed edge in $G^C$.
    \item
    \label{item-2-of-direction of-edges-of-G^C-in-alg}
    For $u\in X_i$, and $v \notin C\cup X_1 \cup X_2\cup \ldots \cup X_i$, if $(u,v)\in E_G$ then $u\rightarrow v$ is a directed edge in $G^C$.
    \item 
    \label{item-3-of-direction-of-edges-of-G^C-in-alg}
    For $u,v\in X_i$, for $1\leq i \leq l$, if $(u,v)\in E_G$ then $u-v$ is an undirected edge in $G^C$.
    \item 
    \label{item-3b-of-direction-of-edges-of-G^C-in-alg}
    For $u,v\in C$, $u-v$ is an undirected edge in $G^C$.
    
    \item 
    \label{item-4-of-direction-of-edges-of-G^C-in-alg}
    For any $X_i\in \mathcal{L}$, every undirected connected component of
    $G[X_i]$ is an element of $\mathcal{C}_G(C)$.
\end{enumerate}
\end{observation}

The following lemma encapsulates the correctness of \cref{alg:AMO-Union-K}.

 \begin{lemma}
 \label{lem:output-of-Alg-1}
 Let $G$ be a UCCG, $C$ be a maximal clique of $G$, and $\mathcal{K}$ be the known background knowledge about $G$. For the input $G,C,$ and $\mathcal{K}$, if $G^C$ is not $\mathcal{K}$-consistent \cref{alg:AMO-Union-K} outputs $(0,\mathcal{C}_G(C))$ on \cref{alg:return}, else it returns $(1,\mathcal{C}_G(C))$ on \cref{alg:return}.
 \end{lemma}

For any maximal clique $C$ of $G$ such that $G^C$ is $\mathcal{K}$-consistent, to compute the size of the set of $\mathcal{K}$-consistent AMOs of $G$ that are canonically represented by $C$, we further partition the set based on the different permutations of $C$. The simple \cref{lem:uniqueness-of-permutation-of-clique} below assists us in mapping each AMO of the set to a unique permutation $\pi(C)$ of $C$.
\begin{observation}
\label{lem:uniqueness-of-permutation-of-clique}
Let $G$ be a UCCG, and $\alpha$ an AMO of $G$ that is represented by a maximal clique $C$ of $G$. Then, there exists a unique permutation $\pi(C)$ of $C$ that represents $\alpha$. 
\end{observation}

By slightly extending the definition of the canonical representation of an AMO by a clique, we say that an AMO is \textbf{canonically represented by $\pi(C)$} if the AMO is represented by $\pi(C)$, and also canonically represented by the clique $C$. Then, \cref{lem:uniqueness-of-permutation-of-clique} implies that we can partition the set of $\mathcal{K}$-consistent AMOs that are canonically represented by $C$ into $\mathcal{K}$-consistent AMOs that are canonically represented by its permutations $\pi(C)$, i.e., $\{\alpha: \alpha \in \text{AMO}(G,\pi(C),\mathcal{K}) \text{ and } C=C_{\alpha}\}$. More formally,
\begin{equation*}
    |\{\alpha: \alpha \in \text{AMO}(G,\mathcal{K}) \text{ and } C=C_{\alpha}\}|= \sum_{\pi(C)}{|\{\alpha: \alpha \in \text{AMO}(G,\pi(C),\mathcal{K}) \text{ and } C=C_{\alpha}\}|}.
\end{equation*}

To compute the size of $\mathcal{K}$-consistent AMOs of $G$ that are canonically represented by $\pi(C)$,  we first have to go through the necessary and sufficient conditions for a maximal clique $C$ of $G$ to become $C_{\alpha}$, for an AMO $\alpha$ of $G$. 

\begin{lemma}[Claims 1, 2 and 3 of \citet{wienobst2020polynomial}]
\label{lem:identification-of-C-alpha}
Let $G$ be a UCCG. \citet{wienobst2020polynomial} fix a rooted clique tree of $G$ to define $C_{\alpha}$, for any AMO $\alpha$ of $G$. Let
$\mathcal{T}=(T,R)$ be the rooted clique tree (with root $R$) of $G$ on which $C_{\alpha}$ is defined, for each AMO $\alpha$ of $G$. For an AMO $\alpha$ of $G$, and a maximal clique $C$ of $G$, $C=C_{\alpha}$ if, and only if,
\begin{enumerate}
\item There exists an LBFS ordering of $G$ that starts with C, and represents $\alpha$, and
\item If $\pi(C)$ is the permutation of $C$ that represents $\alpha$ (from \cref{lem:uniqueness-of-permutation-of-clique}) then there does not exist any edge $C_i-C_j$ in the path in $T$ from $R$ to
  $C$ such that $\pi(C)$ has a prefix $C_i\cap C_j$.
\end{enumerate}
\end{lemma}
The set FP$(C,\mathcal{T})$ is defined to be the set of such forbidden prefixes
$C_i\cap C_j$.

\begin{definition}[$FP(C, \mathcal{T})$,  Definition 3 of \citet{wienobst2020polynomial}]
\label{def:FP}
Let $G$ be a UCCG, $\mathcal{T}=(T,R)$ a rooted clique tree of $G$, $C$ a
node in $T$ and $R=C_1 - C_2 - \ldots - C_{p-1} - C_p=C$ the unique path from $R$ to $C$ in
$T$. We define the set \emph{FP$(C,\mathcal{T})$ } to contain all sets of the
form $C_i\cap C_{i+1} \subseteq C$, for $1 \leq i < p$.
\end{definition}

Based on \cref{lem:identification-of-C-alpha}, we define $(\mathcal{K},\mathcal{T})$-consistency for a permutation $\pi(C)$ to simplify our computation. This definition is one of the main new ingredients that let us extend the result of \cite{wienobst2020polynomial}.

\begin{definition}[$(\mathcal{K},\mathcal{T})$-consistency of permutations of maximal cliques]
\label{def:background-consistent-pi-C}
Let $G$ be a UCCG, $C$ a maximal clique in $G$, $\pi(C)$ a permutation of $C$, $\mathcal{K}$ be a given background knowledge, and $\mathcal{T}=(T,R)$ a rooted clique tree of $G$ (on which $C_{\alpha}$ is defined). $\pi(C)$ is said to be $(\mathcal{K},\mathcal{T})$-consistent if  (a) $\pi(C)$ is $\mathcal{K}$-consistent, i.e, for any edge $u\rightarrow v \in \mathcal{K}$ such that $u,v\in C$, $u$ occurs before $v$ in $\pi(C)$, and (b) no element of FP$(C,\mathcal{T})$ (\cref{def:FP}) is a prefix of $\pi(C)$.
\end{definition}

If $\pi(C)$ itself is not $\mathcal{K}$-consistent then no $\mathcal{K}$-consistent AMO exists that is represented by $\pi(C)$. Also, if $\pi(C)$ has a prefix in  FP$(C,\mathcal{T})$ then from \cref{lem:uniqueness-of-permutation-of-clique,lem:identification-of-C-alpha}, there does not exist an AMO $\alpha$ of $G$ such that $\alpha$ is represented by $\pi(C)$, and $C=C_{\alpha}$. 
This yields the following observation.
\begin{observation}
\label{lem:formular-for-non-KT-consistent-permutaion}
If $\pi(C)$ is not $(\mathcal{K},\mathcal{T})$-consistent then there exists no $\mathcal{K}$-consistent AMOs of $G$ that can be canonically represented by $\pi(C)$, i.e.,  $|\{\alpha: \alpha \in \text{AMO}(G,\pi(C),\mathcal{K}) \text{ and } C=C_{\alpha}\}|= 0$.
\end{observation}

We thus focus on only those permutations $\pi(C)$ of $C$ that are
$(\mathcal{K},\mathcal{T})$-consistent.  The main ingredient towards this end is
the following recursive formula.
\begin{lemma}
\label{lem:formular-for-KT-consistent-permutaion}
Let $G^C$ be $\mathcal{K}$-consistent. Then, for any
$(\mathcal{K},\mathcal{T})$-consistent permutation $\pi(C)$ of $C$, the size of
the set
$\{\alpha: \alpha \in \text{AMO}(G,\pi(C),\mathcal{K}) \text{ and }
C=C_{\alpha}\}$ is $\prod_{H\in \mathcal{C}_G(C)}{\#AMO(H,\mathcal{K}[H])}$.
\end{lemma}
Note that the formula obtained in
\cref{lem:formular-for-KT-consistent-permutaion} depends only upon the clique
$C$ and \emph{not} on the permutation $\pi$ of the nodes of $C$ (as long as
$\pi$ is itself $\mathcal{K}$-consistent)!  This implies immediately that the
number of $\mathcal{K}$-consistent AMOs of $G$ for which $C$ is the {canonical
  representative} is given by multiplying the product
$\prod_{H\in \mathcal{C}_G(C)}{\#AMO(H,\mathcal{K}[H])}$ with
the number of $(\mathcal{K},\mathcal{T})$-consistent permutation of $C$.

This motivates us to count  $(\mathcal{K},\mathcal{T})$-consistent permutations of a maximal clique $C$ of $G$. To count the $(\mathcal{K},\mathcal{T})$-consistent permutations of $C$, we define the following:
\begin{definition}
\label{def:phi-S-R-K}
Let $S$ be a set of vertices, $\mathcal{R}=\{R_1,R_2,\ldots, R_l\}$ such that $R_1\subsetneq R_2 \subsetneq \ldots \subsetneq R_l \subsetneq S$, and $\mathcal{K}\subseteq S\times S$. $\Phi(S,\mathcal{R}, \mathcal{K})$ is the number of $\mathcal{K}$-consistent permutations of $S$ that do not have a prefix in $\mathcal{R}$.
\end{definition}

\cref{lem:partition-of-set-of-bkc-AMOs,lem:formular-for-non-KT-consistent-permutaion},
along with \cref{lem:formular-for-KT-consistent-permutaion} and the discussion
following it, finally give us the following recursion.
 \begin{lemma}
 \label{lem:final-counting-algorithm}
 Let $G$ be a UCCG, $\mathcal{K}$ be a given background knowledge, and
$\mathcal{T}=(T,R)$ a rooted clique tree of $G$ on which $<_{\alpha}$ has been
defined. Then $\#AMO(G,\mathcal{K})$ equals
\begin{equation*}
\sum_{C}
  {\Phi(C,FP(C,\mathcal{T}),\mathcal{K}[C]) \times} 
  {\prod_{H\in \mathcal{C}_G(C)}{\#AMO(H,\mathcal{K}[H])}},
\end{equation*}
where the sum is over those $C$ for which $G^C$ is $\mathcal{K}$-consistent.
 \end{lemma}
 \cref{lem:final-counting-algorithm} solves our counting problem. The precondition of $\Phi(S,\mathcal{R}, \mathcal{K})$ in \cref{def:phi-S-R-K}  that $\mathcal{R}=\{R_1,R_2,\ldots,R_l\}$ has the property
$R_1\subsetneq R_2 \subsetneq \ldots \subsetneq R_l \subsetneq S$ is satisfied
at the beginning of the recursion, i.e. when $\mathcal{R} = FP(C, \mathcal{T})$,
by \cref{lem:elements-of-FP-is-arranged-as subsets}, which is a consequence of the
standard clique intersection property of clique trees of chordal graphs.
\begin{lemma}[\cite{wienobst2020polynomial}, Lemma 7]
\label{lem:elements-of-FP-is-arranged-as subsets}
We can order the elements of $FP(C,\mathcal{T})$ as $X_1\subsetneq X_2\subsetneq \ldots \subsetneq X_l \subsetneq C$.
\end{lemma}
The precondition of $\Phi(S,\mathcal{R}, \mathcal{K})$ is preserved throughout the recursion described in the lemma. We now give a recursive method to compute
$\Phi(S,\mathcal{R}, \mathcal{K})$. 
\begin{lemma}
\label{lem:phi-computaion}
Let $S$ be a clique, and $\mathcal{K}\subseteq S\times S$ be a set of directed
edges. Let $R=\{R_1,R_2,\ldots , R_l\}$ where $l\geq 1$ be such that
$R_1\subsetneq R_2 \subsetneq \ldots \subsetneq R_l \subsetneq S$.
Then,
\begin{enumerate}
\item \label{item-3-of-phi-computaion}
  $\Phi(S, \emptyset,\mathcal{K})= \frac{|S|!}{|V_{\mathcal{K}}|!} \times
  \Psi(V_{\mathcal{K}}, \mathcal{K})$, where
  $\Psi(V_{\mathcal{K}}, \mathcal{K})$ is the number of $\mathcal{K}$-consistent
  permutations of vertices in $V_{\mathcal{K}}$ ($V_{\mathcal{K}}$ is the set of
  end points of edges in $\mathcal{K}$). (Example: Suppose $S=\{1,2,3,4,5\}$, and $\mathcal{K}=\{1\rightarrow2, 2\rightarrow 3\}$. Then, $V_{\mathcal{K}}=\{1,2,3\}$. And, there exist only one permutation, $(1,2,3)$, of $V_{\mathcal{K}}$ that is $\mathcal{K}$-consistent, i.e., $\Psi(V_{\mathcal{K}}, \mathcal{K}) = 1$. This implies $\Phi(S, \emptyset,\mathcal{K})=20$.)

\item \label{item-1-of-phi-computaion}
    If there exists an edge $u\rightarrow v\in \mathcal{K}$ such that
    $u\in S\setminus R_l$ and $v\in R_l$, then
    $\Phi(S,R,\mathcal{K})= \Phi(S,R-\{R_l\},\mathcal{K}).$
\item 
  \label{item-2-of-phi-computaion}
  If there does not exist an edge $u \rightarrow v\in \mathcal{K}$ such that
  $u\in S\setminus R_l$ and $v\in R_l$, then $\Phi(S,R,\mathcal{K})=$
  $\Phi(S,R-\{R_l\},\mathcal{K}) -$
  ${\Phi(R_l,R-\{R_l\},\mathcal{K}[R_l])\times \Phi(S\setminus
    R_l,\emptyset,\mathcal{K}[S\setminus R_l])}$.
\end{enumerate}
\end{lemma}
\begin{algorithm}[t]
\caption{Valid-Perm$(S,\mathcal{R},\mathcal{K})$}
\label{alg:valid-permutaion-of-clique}
\SetAlgoLined
\SetKwInOut{KwIn}{Input}
\SetKwInOut{KwOut}{Output}
\SetKwFunction{Counting-Permutations}{Counting-Permutations}
\KwIn{A clique $S$, $\mathcal{R}=\{R_1,R_2,\ldots,R_l\}$ such that $R_1\subsetneq R_2 \subsetneq \ldots \subsetneq R_l\subsetneq S$,  and background knowledge $\mathcal{K}\subseteq S\times S$.}
\KwOut{ $\Phi(S,\mathcal{R},\mathcal{K})$.
}

\label{alg-vp:R-is-empty} \If{$\mathcal{R}=\emptyset $} 
{

\KwRet
     $\frac{|S|!}{|{V_{\mathcal{K}}|!}} \cdot {\Psi(V_{\mathcal{K}},
       \mathcal{K})}$\label{alg-vp:return-K-empty} \label{alg-vp:sum-init-R-empty}

   }\label{alg-vp:R-is-empty-end}

\label{alg-vp:sum-init-R-not-empty}
$\texttt{sum}\leftarrow$ Valid-Perm$(S,\mathcal{R}-\{R_l\},\mathcal{K})$

\label{alg-vp:no-edge-from-Rl} \If{$\{(u,v):u\rightarrow v\in \mathcal{K}, v\in R_l$ and $u\notin R_l\}\neq \emptyset$}
{
\KwRet $\texttt{sum}$\label{alg-vp:return2}
}\label{alg-vp:no-edge-from-Rl-end}

\KwRet $\texttt{sum} - \text{Valid-Perm}(R_l,\mathcal{R}-\{R_l\},\mathcal{K}[R_l]) \times \text{Valid-Perm}(S\setminus R_l,\emptyset, \mathcal{K}[S\setminus R_l])$\label{alg-vp:final-return}
\end{algorithm}
 \begin{proof}[Proof of \cref{lem:phi-computaion}]
  Proofs of \cref{item-1-of-phi-computaion,item-2-of-phi-computaion} follow
  easily from the definition of the $\Phi$ function
  (\cref{lem:final-counting-algorithm}), and are similar in spirit to the
  corresponding results of \citet{wienobst2020polynomial} in the setting of no
  background knowledge.  We provide the details of these proofs in the
  supplementary material and focus here on proving \cref{item-3-of-phi-computaion}. If
  $\mathcal{R}=\emptyset$ then $\Phi(S,\mathcal{R}, \mathcal{K})$ is the number
  of $\mathcal{K}$-consistent permutations of $S$. There are
  $\frac{|S|!}{|V_{\mathcal{K}}|!}$ permutations of $S$ consistent with any given ordering of vertices in $V_{\mathcal{K}}$. The total number of
  $\mathcal{K}$-consistent permutations of the vertices in $V_{\mathcal{K}}$ is
  $\Psi(V_{\mathcal{K}},\mathcal{K})$. Therefore, the number of
  $\mathcal{K}$-consistent permutations of $S$ equals
  $\frac{|S|!}{|V_{\mathcal{K}}|!} \times \Psi(V_{\mathcal{K}}, \mathcal{K})$.
\end{proof}
\Cref{alg:valid-permutaion-of-clique} implements \cref{lem:phi-computaion} to
compute $\Phi(S,\mathcal{R},\mathcal{K})$, and its correctness given below, is
an easy consequence of \cref{lem:phi-computaion}.
\begin{observation}
\label{lem:alg-valid-perm-correctness}
For input $S, \mathcal{R} = \{R_1,R_2,\ldots, R_l\},$ and $\mathcal{K}$, where
$R_1\subsetneq R_2 \subsetneq \ldots \subsetneq R_l\subsetneq S$, and
$\mathcal{K}\subseteq S\times S$, \cref{alg:valid-permutaion-of-clique} returns
$\Phi(S,\mathcal{R},\mathcal{K})$.
\end{observation}

We now construct \cref{alg:Count-AMO} that computes $\#$AMO$(G,\mathcal{K})$. \Cref{alg:Count-AMO} evaluates this formula, utilizing memoization to avoid recomputations. 

\begin{algorithm}[t]
  \caption{$\texttt{count}(G, \mathcal{K}, \texttt{memo})$ (modification of an algorithm of
    \cite{wienobst2020polynomial})}
\label{alg:Count-AMO}
\SetAlgoLined
\SetKwInOut{KwIn}{Input}
\SetKwInOut{KwOut}{Output}
\SetKwFunction{Count-AMO}{Count-AMO}

    \KwIn{A UCCG $G$, background knowledge $\mathcal{K}\subseteq E_G$.}
    \KwOut{$\#$AMO$(G,\mathcal{K})$.}
    \SetKwProg{Fn}{function}{}

\If{$G \in \texttt{\textup{memo}}$} {\KwRet $\texttt{memo}[G]$}\label{alg-count:G-in-memory}

$\mathcal{T}=(T,R)\leftarrow  \text{ a rooted clique tree of $G$}$\label{alg-count:clique-tree-construction}

\If{$R=V_G$\label{alg:count:G-is-a-clique}}{

$\texttt{memo}[G]=$\label{alg-count:memo-update}
$\Phi(V,\emptyset,\mathcal{K})$
 
 \KwRet $\texttt{memo}[G]$
 }\label{alg-count:if-G-is-a-clique-end}

 $\texttt{sum} \leftarrow 0$\label{alg-count:sum-initialization}

 $Q\leftarrow \text{queue with single element $R$}$ \label{alg-count:Queue-construction}

 \While{$Q \text{ is not empty \label{alg-count:while-queue-is-not-empty}}$}{
 
 $C\leftarrow pop(Q)$\label{alg-count:clique-popped}\\

 push($Q$, children($C$))\label{alg-count:queue-pushed}\\

 $(\texttt{flag},\mathcal{L})\leftarrow$LBFS$(G,C,\mathcal{K})$\label{alg-count:LBFS}

 \If{$\texttt{\textup{flag}} = 1$ \label{alg-count:flag-is-1}}{

 $\texttt{prod} \leftarrow 1$\label{alg-count:prod-init}

\ForEach{$H\in \mathcal{L}$\label{alg-count:foreach-start}}{

$\texttt{prod} = \texttt{prod} \times \texttt{count}(G[H], \mathcal{K}[H], \texttt{memo})$\label{alg-count:prod-multiplication}
}\label{alg-count:foreach-end}

$\texttt{sum} = \texttt{sum}
+  \texttt{prod} \times  \Phi(C,  \text{FP}(C,\mathcal{T}),\mathcal{K}[C])$\label{alg-count:sum-update} \\
}\label{alg-count:if-end}
}\label{alg-count:while-end}

$\texttt{memo}[G]= \texttt{sum}$\label{alg-count:memory-update}

 \KwRet $\texttt{sum}$\label{alg-count:final-return}
\end{algorithm}

\begin{theorem}\label{lem:main-algorithm-correct}
For a UCCG $G$ and background knowledge $\mathcal{K}$, \cref{alg:Count-AMO} returns $\#$AMO$(G,\mathcal{K})$.
\end{theorem}
Proof of \cref{lem:main-algorithm-correct}:
We first fix a
clique tree $\mathcal{T}=(T,R)$ (at line \ref{alg-count:clique-tree-construction}) on which we define $<_{\alpha}$. Lines \ref{alg:count:G-is-a-clique}-\ref{alg-count:if-G-is-a-clique-end} deals with the base case when $G$ is a clique. If $G$ is not a clique, \cref{alg:Count-AMO} follows \cref{lem:final-counting-algorithm}.  The full detail is given in Supplementary Material due to lack of space.

 \section{Time Complexity Analysis}
\label{sec:time-complexity}
In this section, we analyze the run time of \cref{alg:Count-AMO}.  The proof of
the following observation, which shows that despite our modifications,
\cref{alg:AMO-Union-K} still runs in linear time, is given in the Supplementary
Material.

\begin{observation}
\label{thm:LBFS-alg}
For a UCCG $G$, a maximal clique $C$ of $G$, and background knowledge $\mathcal{K}$, \cref{alg:AMO-Union-K} runs in linear time $O(|V_G| +|E_G|)$.
\end{observation}

Similar to Wien{\"o}bst et al.'s $\texttt{count}$ function, our $\texttt{count}$ function
(\cref{alg:Count-AMO}) is also recursively called at most $2|\Pi(G)|-1|$ times,
where $\Pi(G)$ is the set of maximal cliques of $G$.  Our approach to compute
the background aware version of $\Phi$ (\cref{alg:valid-permutaion-of-clique})
is similar to that of \citet{wienobst2020polynomial}, and the difference in time
complexity comes from the high time complexity of computation of
$\Phi(S,\emptyset,\mathcal{K})$ at \cref{item-3-of-phi-computaion} (it is $O(1)$
for $\mathcal{K}=\emptyset$, which is the setting considered by
\citet{wienobst2020polynomial}).  Proof of the claims below can be found in the
Supplementary Material.

\begin{proposition}
\label{prop:number-of-times-alg-2-called}
Let $G$ be a UCCG, and $\mathcal{K}$ be the known background knowledge about $G$. The number of distinct UCCG explored by the \texttt{count} function (as defined in \cref{alg:Count-AMO}) is bounded by $2|\Pi(G)|-1$.
\end{proposition}

\begin{lemma}
\label{lem:time-complexity-of-Phi}
For input $S$, $\mathcal{R}=\{R_1,R_2,\ldots, R_l\}$, and $\mathcal{K}$,
\cref{alg:valid-permutaion-of-clique} can be implemented using memoization to
use $O(k! \cdot k^2 \cdot |\Pi(G)|^2)$ arithmetic operations, where $k$ is the
max-clique knowledge of $\mathcal{K}$ (assuming factorials of integers from $1$
to $|V_G|$ are available for free).
\end{lemma}

\begin{theorem}[Final runtime bound of \cref{alg:Count-AMO}]
\label{thm:ACMO-Counting-BK}
For a UCCG $G$, and background knowledge $\mathcal{K}$, \cref{alg:Count-AMO} runs in time $O(k!k^2n^4)$, more precisely $O({k!k^2 \cdot |\Pi(G)|}^4)$, where $n$ is the number of nodes in $G$, and $k$ is the max-clique knowledge of $\mathcal{K}$.
\end{theorem}

\begin{proof}[Proof of \cref{thm:main-result}]
  Together, \cref{lem:main-algorithm-correct,thm:ACMO-Counting-BK} prove our
  main result, \cref{thm:main-result}.
\end{proof}
 \section{Experimental evaluation}
\label{sec:experimental_evaluation}
In this section, we evaluate the performance of \cref{alg:Count-AMO} on a synthetic dataset. For each $n\in \{500, 510, 520, \ldots, 1000\}$, we construct 50 random chordal graphs with $n$ nodes, and for each
$k \in \{5,6,\ldots, 13\}$, we construct a set of background knowledge edges with $k$ as its max clique knowledge value. We then
measure the running time of \cref{alg:Count-AMO} for each of these (graph,
background knowledge) pairs, and take the mean running time over all such pairs
with the same value of $n$ and $k$.  Further details about the construction of
these instances can be found in the Supplementary Material.

\paragraph{Validating the run-time bound} To validate the $O(k! k^2 n^4)$
run-time bound established in \cref{sec:time-complexity}, we draw log-log plots
of the mean run-time $T$ against the size $n$ of the graph, for each fixed value
of $k$ (\cref{fig:comp-plot-k-vs-logT-5-13}).  As predicted by the polynomial
(in $n$) run-time bound in our theoretical result, we get, for each value of the
parameter $k$, a roughly linear log-log plot.

\begin{figure}[t]
  \centering   \includegraphics[width=0.48\textwidth]{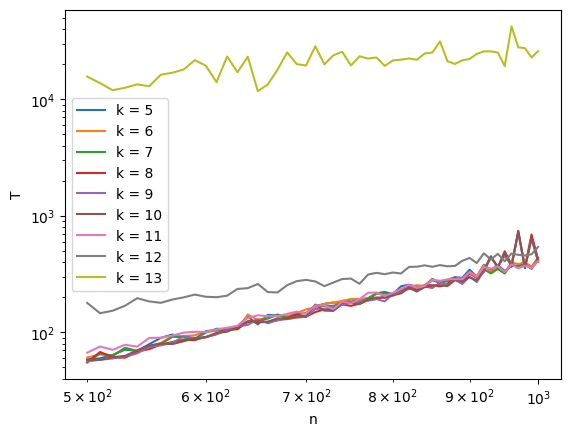}
  \caption{$\log$ vs. $\log$ plot of $T$ vs. $n$ \label{fig:comp-plot-k-vs-logT-5-13}}
\end{figure}

\paragraph{The intercept of the log-log plot} While the plots in
\cref{fig:comp-plot-k-vs-logT-5-13} for $k\in \{5, 6, 7, 8, 9, 10, 11\}$ are
quite close to each other, the separation of the plots for $k=12$ and $k=13$ is
much larger. The reason behind it is the actual time complexity of
\cref{alg:Count-AMO}  can roughly be
bounded as $\log T \leq \log a + \log (k!\cdot k^2 + b) + 4\log n$, where $a$ and $b$ are constants independent of $n$ and $k$.
The above observation then shows that until about $k \approxeq 11$, all the plots
have intercepts close to each other, as the contribution of terms involving $a,b$ and $n$ dominates that of $k!\cdot k^2$ for
small values of $k$.

\paragraph{Effect of the size of the background knowledge} An important feature
of our analysis of \cref{alg:Count-AMO} is that its run-time bound does not
depend directly upon the actual size of the background knowledge.  To validate
this, we conduct the following experiment: we fix a chordal graph of size $n$
and the max-clique knowledge value $k$, and then construct two different sets
$\mathcal{K}_1$ and $\mathcal{K}_2$ of background knowledge edges, of different
sizes such that both have the same $k$ value (the details of the construction
are given in the Supplementary Material).  In \cref{tab:1}, $T_1$, $T_2$ are the
running times of \cref{alg:Count-AMO} with background knowledge $\mathcal{K}_1$
and $\mathcal{K}_2$ respectively.
\begin{table}[h]
  \centering 
  \begin{tabular}{ |c|c|c|c|c|c| }
    \hline
    $n$ & $k$ &$|\mathcal{K}_1|$ & $|\mathcal{K}_2|$ & $T_1$ & $T_2$ \\ 
    \hline
    1000 & 10 & 65 & 116 & 370 & 353 \\  
    \hline
    1000 & 10 & 75 & 137 & 346 & 338 \\ 
    \hline
    1100 & 11 & 55 & 121 & 467 & 455 \\ 
    \hline
    1100 & 11 & 51 & 104 & 460 & 453 \\
    \hline
  \end{tabular}
  \caption{Exploring runtime dependence on the number of background knowledge
    edges~\label{tab:1}}
\end{table}
The table confirms the expectation that when the graph and $k$ are fixed, the
running time does not increase much when the size of the background-knowledge
increases.  More detailed data and discussion of this phenomenon are given in
the Supplementary Material.

 \section{Conclusion}
\label{sec:conclusion-open-problems}

Our main result shows that the max-clique-knowledge parameter we introduce plays
an important role in the algorithmic complexity of counting Markov equivalent
DAGs under background knowledge constraint. In particular, it leads to a
polynomial time algorithm in the special case of graphs of bounded
maximum-clique size. Note that an algorithm that runs in polynomial time in the
general case is precluded by the $\#$P-hardness result of
\cite{wienobst2020polynomial} (unless P = NP). However, the optimal dependence
of the run time on the max-clique-knowledge parameter is an interesting open
problem left open by our work.

\section{Acknowledgment}
We acknowledge the support from the Department of Atomic Energy, Government of
India, under project no. RTI4001.  We want to thank Piyush Srivastava for his
invaluable suggestions, discussions, and help in the completion of this
paper. We thank all the anonymous reviewers for multiple useful suggestions which helped
in improving the presentation of this paper.

\clearpage
\appendix
\section*{Supplementary Material}
\section{Graph terminology}
\label{app:graph-terminology}
\textbf{Graphs and graph unions.} 
We mostly follow the graph theory terminology
of \cite{andersson1997characterization}.
A graph is \emph{undirected} if all its edges are
undirected, \emph{directed} if all its edges are \emph{directed}, and
\emph{partially directed} if it contains both directed and undirected edges. A
directed graph that has no directed cycle is called a \emph{directed acyclic
  graph} (DAG). A generalization of this idea is that of \emph{chain graphs}: a
partially directed graph is called a chain graph if it has no cycle which
contains (i) at least one directed edge, and (ii) in which all directed edges
are directed in the same direction as one moves  along the cycle. We denote the
neighbors of a vertex $v$ in a graph $G$ as $N_G(v)$, and an induced subgraph of
$G$ on a set $X\subseteq V$ is denoted as $G[X]$. 
The
\emph{graph union} (which we just call ``union'') $G_1\cup G_2$ includes
vertices and edges present in any one of $G_1$ or $G_2$, i.e., $V_{G_1\cup G_2} = V_{G_1} \cup V_{G_2}$, and $E_{G_1\cup G_2} = E_{G_1} \cup E_{G_2}$.  The \emph{skeleton} of
a partially directed graph $G$ is an undirected version of $G$ which we get by
ignoring the direction of all the edges: in particular, note that the {skeleton} of
a partially directed graph is also the graph union of \emph{all} the partially
directed graphs with the same skeleton as $G$.  A \emph{v-structure} (or
\emph{unshielded collider}) in a partially directed graph $G$ is an ordered
triple of vertices $(a,b,c)$ of $G$ which induce the subgraph
$a\rightarrow b \leftarrow c$ in $G$.

\textbf{Cliques, separators, chordal graphs, and UCCG.} A \emph{clique} is a set of pairwise adjacent vertices for a graph. For a graph $G$,  $u$ and $v$ are said to be pairwise adjacent to each other if $(u,v),(v,u) \in E_G$, i.e, $u-v\in E_G$.
For an undirected graph $G$, a set $S\subset V_G$ is an $x$-$y$ \emph{separator} for two non-adjacent vertices $x$ and $y$ if $x$ and $y$ are in two different undirected connected components of $G[V_G\setminus S]$.  $S$ is said to be a \emph{minimal $x$-$y$ separator} if no proper subset of $S$ separates $x$ and $y$.  A set $S$ is said to be a \emph{minimal vertex separator} if there exist vertices $x$ and $y$ for which $S$ is a minimal $x$-$y$ separator.\footnote{\citet{wienobst2020polynomial} refer to these objects as \emph{minimal separators}, but we follow here the terminology of \cite[Section 2.2]{blair1993introduction} for consistency.} An undirected graph $G$ is \emph{chordal} if, for any cycle of length $4$ or more of $G$, there exist two non-adjacent vertices of the cycle which are adjacent in
$G$.  We refer to an undirected connected chordal graph by the abbreviation
\emph{UCCG}.

\textbf{Clique trees.} A \emph{rooted clique tree} of a UCCG $G$ is a
tuple $\mathcal{T}=(T, R)$, where $T$ is a rooted tree (rooted at the node $R$) whose nodes are the maximal cliques of $G$, and which is such that the set $\{C: v\in C\}$ is connected in $T$, for all $v\in V_G$.
Clique trees satisfy the important \emph{clique-intersection property}: if $C_1,C_2,C\in V_T$ and $C$ is on the (unique) path between $C_1$ and $C_2$ in the tree $T$, then $C_1\cap C_2 \subset C$~(see, e.g., \citet{blair1993introduction}).  Further, a set $S\subset V_G$ is a minimal vertex separator in $G$ if, and only if, there are two adjacent nodes $C_1, C_2\in V_T$ such that $C_1 \cap C_2 = S$ \citep[Theorem 4.3]{blair1993introduction}.  A clique tree for a UCCG $G$ can be constructed in polynomial time.  For more details on the above results on chordal graphs and clique trees, we refer to the survey of \citet{blair1993introduction}. \section{Proofs omitted from Section~\ref{sec:main-result}}
\label{sec:app-main-result}
\begin{observation}
\label{obs:if-G-is-not-K-consistent}
If $G$ is not $\mathcal{K}$-consistent then $\#AMO(G,\mathcal{K})=0$. 
\end{observation}
\begin{proof}
   If $G$ is not $\mathcal{K}$-consistent then  there exists an edge $u\rightarrow v \in \mathcal{K}$ such that $v\rightarrow u \in E_G$, and all the AMOs of $G$ have the edge $v\rightarrow u$. This shows that no AMO of $G$ is $\mathcal{K}$-consistent.
\end{proof}
We can verify in polynomial time that $G$ is $\mathcal{K}$-consistent or not, by checking the existence of an edge $u\rightarrow v \in \mathcal{K}$ for which $v\rightarrow u$ is a directed edge in $G$. This is why for further discussion we assume that $G$ is $\mathcal{K}$-consistent.

\cref{lem:Counting-AMO-reduction} is a direct consequence of the following
lemma.

\begin{lemma}
\label{lem:AMO-reduction}
Let $G$ be an MEC consistent with a given background knowledge $\mathcal{K}$, and let $G_d$ be the directed subgraph of $G$. Then, $\alpha\in AMO(G,\mathcal{K})$ if, and only if, (i) for each undirected chordal component $H$ of $G$, $\alpha[V_H]$ is a $\mathcal{K}$-consistent AMO of $H$ and (ii) $\alpha$ is a union of $G_d$ and $\bigcup_{H}{\alpha[V_H]}$.
\end{lemma}

\begin{proof}
If $G$ is $\mathcal{K}$-consistent then each directed edge of $G$ is $\mathcal{K}$-consistent, i.e., $G_d$ is $\mathcal{K}$-consistent.
  \cite{andersson1997characterization} show that for an MEC $G$, an AMO of $G$ can be constructed by choosing an AMO from each one of the chordal components of  $G$ and taking the union of the directed subgraph of $G$ and the chosen AMOs of  the chordal components.  For every undirected connected chordal component $H$  of $G$, let us pick a $\mathcal{K}$-consistent AMO  of $H$. Then, the union of all the picked AMOs and $G_d$ is a $\mathcal{K}$-consistent AMO of  $G$.  Also if $\alpha$ is an AMO of $G$ then for all undirected connected chordal components $H$ of $G$, $\alpha[V_H]$ is $\mathcal{K}$-consistent. And $\alpha$ must be a union of $G_d$ and the union of all the $\alpha[V_H]$.  This proves \cref{lem:AMO-reduction}.
\end{proof}

 \section{Definition of $C_{\alpha}$}
\label{supplementary-sec:definition-C-alpha}
Here, we define $C_{\alpha}$.  We start with defining a few preliminary terms  flower, $<_{\mathcal{T}}$, and $<_{\alpha}$ that we use to define $C_{\alpha}$.

\begin{definition}[Flowers and bouquets, \cite{wienobst2020polynomial},
  Definition 4]
\label{def:flower}
Let $G$ be a UCCG. An \emph{$S$-flower} for a minimal vertex separator $S$ of
$G$ is a maximal subset $F$ of the set of maximal cliques of $G$ containing $S$
such that $\bigcup_{C\in F}C$ is connected in the induced subgraph
$G[V\setminus S]$. The \emph{bouquet} $B(S)$ of a minimal separator $S$ is the set of
all $S$-flowers.
\end{definition}

\begin{definition}[The $<_{\mathcal{T}}$ order for a rooted clique tree
  $\mathcal{T}$, Section 5 of \cite{wienobst2020polynomial}]
\label{def:ordering-of-flowers}
Let $G$ be a UCCG, $S$ be a minimal vertex separator of $G$, $F_1,F_2\in B(S)$, and $\mathcal{T}=(T,R)$ be a rooted clique tree of $G$. \emph{$F_1 <_{\mathcal{T}} F_2$} if $F_1$ contains a node on the unique path from $R$ to $F_2$. 
\end{definition}

\begin{definition}[The $<_{\alpha}$ order for an AMO
  $\alpha$, Section 5 of \cite{wienobst2020polynomial}]
\label{def:ordering-of-cliques}
Let $G$ be a UCCG, $\mathcal{T}=(T,R)$ be a rooted clique tree of $G$, and
$\alpha$ be an AMO of $G$.  We use $<_{\mathcal{T}}$ to define a partial order
\emph{$<_{\alpha}$} on the set of maximal cliques that represent $\alpha$, as
follows: $C_1 <_{\alpha} C_2$ if, and only if, $(i)$ $C_1\cap C_2 =S$ is a
minimal vertex separator, $(ii)$ $C_1$ and $C_2$ are elements of distinct
$S$-flowers $F_1,F_2\in B(S)$, respectively, and $(iii)$ $F_1 <_T F_2$.
\end{definition}

The following result of \cite{wienobst2020polynomial} establishes the requisite property of the ordering $<_{\alpha}$.
\begin{lemma}[\cite{wienobst2020polynomial}, Claim 1]
\label{lem:C-alpha}
  Let $G$ be a UCCG, $\alpha$ an AMO of $G$, and $\mathcal{T}=(T,R)$ a rooted
  clique tree of $G$.  Consider the order $<_{\alpha}$ defined on the maximal
  cliques $\mathcal{T}$. Then, there always exists a unique least maximal
  clique with respect to $<_{\alpha}$. 
\end{lemma} \section{Proofs omitted from Section~\ref{sec:reduction-of-the-problem}}
\label{sec:app-lbfs}

\begin{proof}[Proof of \cref{lem:partition-of-set-of-bkc-AMOs}]
  The claim follows from the fact that for each AMO $\alpha$ of $G$,
  $C_{\alpha}$ was canonically chosen from the set of maximal cliques representing
  $\alpha$.
\end{proof}

\begin{proof}[Proof of \cref{lem:direction of-edges-of-GC-in-alg}]
Proof of \cref{item-1-of-direction of-edges-of-G^C-in-alg}:
For any edge $u-v \in E_G$, if $u\in C$ and $v\notin C$ then for any AMO that is represented by an LBFS ordering that starts with $C$, $u\rightarrow v$ is a directed edge in the AMO (see ``Representation of an AMO'' of \cref{sec:preliminaries}). This further implies for any edge $u-v \in E_G$, if $u\in C$ and $v\notin C$ then $u\rightarrow v$ is a directed edge in $G^C$, as $G^C$ is the union of all the AMOs of $G$ that can be represented by an LBFS ordering that starts with $C$. This proves \cref{item-1-of-direction of-edges-of-G^C-in-alg}.

Proof of \cref{item-2-of-direction of-edges-of-G^C-in-alg}:
We prove \cref{item-2-of-direction of-edges-of-G^C-in-alg} of \cref{lem:direction of-edges-of-GC-in-alg} using induction on the size of $\mathcal{L}$. 

    \textbf{Base Case: $|\mathcal{L}|=0$.} In this case, \cref{item-2-of-direction of-edges-of-G^C-in-alg} of \cref{lem:direction of-edges-of-GC-in-alg} is vacuously true. 

    Let \cref{item-2-of-direction of-edges-of-G^C-in-alg} is true when $|\mathcal{L}|=l\geq 0$. We show that \cref{item-2-of-direction of-edges-of-G^C-in-alg} is true even for $|\mathcal{L}|=l+1$.

    Let at some iteration of \cref{alg:AMO-Union-K}, $\mathcal{L}=\{X_1,X_2,\ldots,X_l,X_{l+1}\}$. From the induction hypothesis, for an edge $u-v \in E_G$, if $u\in X_{i}$ and $v\notin C\cup X_1\cup X_2 \cup \ldots \cup X_i$ then $u\rightarrow v \in G^C$, when $i\leq l$. Let there exists an edge $u-v \in  E_G$ such that $u\in X_{l+1}$ and $v\notin C\cup X_1\cup X_2 \cup \ldots \cup X_{l+1}$. This means there must exist a vertex $x\in C\cup X_1\cup X_2 \cup \ldots \cup X_l$ for which $x-u\in E_G$ and $v-u \notin E_G$, due to which $u$ and $v$ moves to two different sets in $\mathcal{S}$, because initially $u$ and $v$ are in the same set $V\setminus C$. From the induction hypothesis, $x\rightarrow u \in G^C$. This implies all the AMOs that are represented by $C$ have edge $x\rightarrow u$. This further implies all the AMOs that are represented by $C$ have edge $u\rightarrow v$, as  $v\rightarrow u$ creates an immorality $x\rightarrow u \leftarrow v$ (from the definition of AMO, there cannot be a v-structure in an AMO of $G$). This proves \cref{item-2-of-direction of-edges-of-G^C-in-alg}.

Proof of \cref{item-3-of-direction-of-edges-of-G^C-in-alg}:
    Suppose $u,v \in X_{i}$, and $u-v\in E_G$. Let there exists an AMO $\alpha$ that is represented by $C$ and has the edge $u\rightarrow v$. Let $\tau_1$ be an LBFS ordering of $G$ that starts with $C$, and represents $\alpha$. We can construct another LBFS ordering $\tau_2$ that also starts with $C$ such that while picking the vertices of $X_{i}$, we pick $v$ before $u$. The AMO corresponding to this LBFS ordering has the edge $v\rightarrow u$. Since $G^C$ is the union of all the AMOs of $G$ that is represented by  $C$. This implies $G^C$ has the undirected edge $u-v$. This proves \cref{item-3-of-direction-of-edges-of-G^C-in-alg}.

Proof of \cref{item-3b-of-direction-of-edges-of-G^C-in-alg}:
If $u,v \in C$ then $u-v \in E_G$, because $C$ is a clique of $G$. Suppose an AMO $\alpha$ is represented by $C$ and has the  edge $u \rightarrow v$. Then, $\alpha$ must be  represented by an LBFS ordering $\tau_1$ that starts with a permutation $\pi_1(C)$ of $C$ such that $u$ comes before $v$ in $\pi_1(C)$. Let $\pi_2(C)$ be a permutation of $C$ such that $v$ comes before $u$ in $\pi_2(C)$. Let us construct an LBFS ordering $\tau_2$ by replacing $\pi_1(C)$ with $\pi_2(C)$ in $\tau_1$. Let $\beta$ be the AMO represented by the LBFS ordering $\tau_2$. $\beta$ also represented by $C$ and has the edge $v\rightarrow u$. Since $G^C$ is the union of all the AMOs of $G$ that is represented by $C$, this implies $u-v$ is an undirected edge in $G^C$, because $u\rightarrow v\in \alpha$, and $v\rightarrow u\in \beta$, and both are represented by $C$. This proves \cref{item-3b-of-direction-of-edges-of-G^C-in-alg}.

Proof of \cref{item-4-of-direction-of-edges-of-G^C-in-alg}: From \cref{item-1-of-direction of-edges-of-G^C-in-alg,item-2-of-direction of-edges-of-G^C-in-alg}, all the edges with only one endpoint in $X_l$ are directed in $G^C$. And, from \cref{item-3-of-direction-of-edges-of-G^C-in-alg}, all the edges with both of the endpoints in $X_l$ are undirected in $G^C$. This further implies all the undirected connected components of $G[X_i]$ are undirected connected components of $G^C$.  
\end{proof}

 \begin{proof}[Proof of \cref{lem:output-of-Alg-1}]
 At first, we want to recall that
\cref{alg:AMO-Union-K} is background aware version of the modified LBFS algorithm of \cite{wienobst2020polynomial}. As discussed in the main paper, we do not change any line from the LBFS algorithm of \cite{wienobst2020polynomial}; the only modifications we do to their LBFS algorithm are (a) introduction of \texttt{flag}, at \cref{alg:flag-init}, which is used to check the $\mathcal{K}$-consistency of $G^C$, (b) lines \ref{alg:mod-if}-\ref{alg:mod-if-end}, which is used to update the value of \texttt{flag}, and (c) We also output the value of \texttt{flag} with $C_G(C)$. The output of \cref{alg:AMO-Union-K} has 2 components. The first component is the value of \texttt{flag}, and the second value is the value returned by the LBFS algorithm of \cite{wienobst2020polynomial}, which is $C_G(C)$. Thus, the only thing we need to verify is that the first component of our output, i.e., the value of \texttt{flag}, is  1 if $G^C$ is $\mathcal{K}$-consistent, and 0 if $G^C$ is not $\mathcal{K}$-consistent, which is equivalent to show that the value of \texttt{flag} returned by the algorithm is 0 if, and only if, $G^C$ is not $\mathcal{K}$-consistent (since the value of \texttt{flag} is either 0 or 1).

Suppose $G^C$ is not $\mathcal{K}$-consistent. Then, from the definition of $\mathcal{K}$-consistency of $G^C$ (\cref{def:bc-CGC}), there must exist an edge $v\rightarrow u$ in $G^C$ such that $u\rightarrow v \in \mathcal{K}$. From \cref{lem:direction of-edges-of-GC-in-alg}, if $v\rightarrow u \in G^C$ then either (a) $v\in C$ and $u\notin C$, or (b) at some iteration, when $\mathcal{L}=\{X_1, X_2, \ldots , X_l\}$, $v\in X_l$ and $u\notin C\cup X_1\cup X_2 \cup \ldots \cup X_l$. In both of the cases,  $v$ must be picked before $u$ at line-\ref{alg:mod1}, as $v$ is present in a set that comes before the set in which $u$ is present.  At the iteration when $v$ is picked,  the algorithm finds the edge $u\rightarrow v$  that obeys the condition stated in line-\ref{alg:mod-if}. This further sets the value of \texttt{flag} to 0. From the construction of the algorithm, once the value of \texttt{flag} sets to 0, it remains at 0. This shows that if $G^C$ is not $\mathcal{K}$-consistent then the value of \texttt{flag} is 0.

Now suppose the \texttt{flag} value returned by the algorithm is 0. At line-\ref{alg:flag-init}, the value of \texttt{flag} is initialized with 1. If 
\texttt{flag} value returned by the algorithm is 0 then there must exist an edge $u\rightarrow v \in \mathcal{K}$ (found at line-\ref{alg:mod-if}), which causes to change the value of \texttt{flag} to 0 at line-\ref{alg:mod-output-empty}. Since $u\rightarrow v$ obeys the condition state at line-\ref{alg:mod-if}, we can certainly say that the iteration when $v$ is picked at the line-\ref{alg:mod1}, $u$ must be neither in $C$ nor in any set of $\mathcal{L}$. And, $v$ must be either in $C$, if $\mathcal{L}=\emptyset$, or in $X_l$, if $\mathcal{L}=\{X_1, X_2, \ldots , X_l \}$ such that $l\geq 1$. From \cref{lem:direction of-edges-of-GC-in-alg}, $v\rightarrow u$ must be a directed edge in $G^C$. This makes $G^C$ inconsistent with $\mathcal{K}$, as $u\rightarrow v \in \mathcal{K}$ (\cref{def:bc-CGC}). This shows if the value of \texttt{flag} returned by the algorithm is 0 then  $G^C$ is not $\mathcal{K}$-consistent. This completes the proof.
\end{proof}

\begin{proof}[Proof of \cref{lem:uniqueness-of-permutation-of-clique}]
If $\alpha$ is a member of AMO$(G,\pi_1(C))$ and AMO$(G,\pi_2(C))$ both, for two different permutations $\pi_1(C)$ and $\pi_2(C)$, then there must exist two vertices $u,v\in C$ such that $u$ comes before $v$ in $\pi_1(C)$, and $v$ comes before $u$ in $\pi_2(C)$. Since $u$ and $v$ are members of the same clique, there exists an edge $u-v \in E_G$. And, since the AMO is a member of both AMO$(G,\pi_1(C))$, and AMO$(G,\pi_2(C))$ it should have both $u\rightarrow v$ and $v\rightarrow u$, which is not possible. This implies there exists a unique permutation $\pi(C)$ of $C$ that represents $\alpha$.
\end{proof}

\begin{proof}[Proof of \cref{lem:identification-of-C-alpha}]
Claims 1 and 2 of \cite{wienobst2020polynomial} translate into  the ``only if'' part of \cref{lem:identification-of-C-alpha}, while Claim 3 of \cite{wienobst2020polynomial} translates into the ``if'' part of \cref{lem:identification-of-C-alpha}.  We give details of the translation below.
{\newcommand{\pre}{<_{\alpha}}
Given an AMO $\alpha$, \cite{wienobst2020polynomial} define a partial order $\prec_\alpha$ (as described above) on the set of maximal cliques that represent $\alpha$.  Claim 1 of \cite{wienobst2020polynomial} shows that there exists a unique maximal clique representing $\alpha$ (which we name as $C_{\alpha}$) such that for any maximal clique $C \neq C_{\alpha}$ that represents $\alpha$, $C_{\alpha} \pre C$.  Claim 2 of \cite{wienobst2020polynomial} then translates immediately into the ``only if'' part of \cref{lem:identification-of-C-alpha}.

For the ``if'' part, we consider any maximal clique $C$ representing $\alpha$ and satisfying both the conditions of \cref{lem:identification-of-C-alpha}.  Suppose, if possible, that $C \neq C_{\alpha}$.  Then, from the above discussion, $C_\alpha \pre C$. The proof of Claim 3 (and the definition of $FP(C, \mathcal{T})$) then shows that if $\pi(C)$ is the permutation of $C$ representing $\alpha$, then there is a prefix $S$ of $\pi(C)$ of the form $C_i \cap C_j$ for some two adjacent cliques on the path from $R$ to $C$ in $T$.  This leads to a contradiction with item 2 of \cref{lem:identification-of-C-alpha}, and hense shows that $C \neq C_\alpha$ is not possible.
}
\end{proof}

For the proof of \cref{lem:formular-for-KT-consistent-permutaion}, we need the
 following observation of \citet{wienobst2020polynomial}.

\begin{observation}[{\citet{wienobst2020polynomial}, Proposition 1}]
  \label{lem:G-C-and-G-pi-C-is-same}
  For each permutation $\pi(C)$ of a maximal clique $C$ of $G$, all edges of $G^{\pi(C)}$ coincide with the edges of $G^{C}$, excluding the edges connecting the vertices in $C$.  In particular, $\mathcal{C}_G(\pi(C))=\mathcal{C}_G(C)$.
\end{observation}

\begin{proof}[Proof of \cref{lem:formular-for-KT-consistent-permutaion}]
    Let $G^C$ be $\mathcal{K}$-consistent, and $\pi(C)$ is a $(\mathcal{K}, \mathcal{T})$-consistent permutation of $C$. We show that the number of $\mathcal{K}$-consistent AMOs of $G$ that are canonically represented by $\pi(C)$, i.e., $|\{\alpha: \alpha \in \text{AMO}(G, \pi(C), \mathcal{K}) \text{ and } C = C_{\alpha} \}|$ equals $\Pi_{H\in \mathcal{C}_G(C)}{\#\text{AMO}(H, \mathcal{K}[H])}$. To prove this, we first show that if $\alpha$ is a $\mathcal{K}$-consistent AMO of $G$ that is canonically represented by $\pi(C)$ then  for any connected component $H$ of $\mathcal{C}_G(C)$, $\alpha[H]$ is $\mathcal{K}[H]$-consistent AMO of $H$. We also show that if we have a $\mathcal{K}[H]$-consistent AMO for each connected component $H$ of $\mathcal{C}_G(C)$ then we can construct a $\mathcal{K}$-consistent AMO of $G$ by combining them. This proves \cref{lem:formular-for-KT-consistent-permutaion}.

    We first show the first part.
    Let $\alpha$ be a $\mathcal{K}$-consistent AMO of $G$ that is canonically represented by $\pi(C)$. Then,  for any connected component $H$ of $\mathcal{C}_G(C)$, $\alpha[H]$ is $\mathcal{K}[H]$-consistent AMO of $H$. Otherwise, if for any connected component $H$ of $\mathcal{C}_G(C)$, $\alpha[H]$ is not $\mathcal{K}[H]$-consistent then there must exist an edge $u\rightarrow v \in \mathcal{K}[H]$ such that $v\rightarrow u \in \alpha[H]$. But, this implies $\alpha$ is not $\mathcal{K}$-consistent either, as if $u\rightarrow v \in \mathcal{K}[H]$ then $u\rightarrow v \in \mathcal{K}$, and if  $v \rightarrow u \in \alpha[H]$ then $v\rightarrow u \in \alpha$.

    We now show the other part. Let $H_1, H_2, \ldots, H_l$ are the undirected connected components of $\mathcal{C}_G(C)$, in the same order as we get as the output of \cref{alg:AMO-Union-K} for input $G, C$, and $\mathcal{K}$. For each $H_i \in \mathcal{C}_G(C)$, let we have a $\mathcal{K}[H_i]$-consistent AMO $D_i$, and ${\tau}_i$ be an LBFS ordering of $D_i$. Then, $\tau = \{\pi(C), {\tau}_1, {\tau}_2, \ldots , {\tau}_l\}$ is a $\mathcal{K}$-consistent LBFS ordering of $G$ starting with $\pi(C)$ that we can get from \cref{alg:AMO-Union-K}. The DAG $\alpha$ represented by $\tau$ is an AMO of $G$ represented by $\pi(C)$. And, since $\pi(C)$ is $(\mathcal{K}, \mathcal{T})$-consistent, from \cref{lem:identification-of-C-alpha}, $C = C_{\alpha}$. This implies $\alpha$ is a $\mathcal{K}$-consistent AMO of $G$ and is canonically represented by $\pi(C)$. This gives us a one-to-one mapping between the set of $\mathcal{K}$-consistent AMOs of $G$ that is canonically represented by $\pi(C)$, and the $\mathcal{K}[H]$-consistent AMOs of the connected components $H$ of $\mathcal{C}_G(C)$, which further implies the equality between the size of $\mathcal{K}$-consistent AMOs of $G$ that is canonically represented by $\pi(C)$ and $\Pi_{H\in \mathcal{C}_G(C)}{\#\text{AMO}(H, \mathcal{K}[H])}$.  This completes our proof.
\end{proof}

\begin{proof}[Proof of \cref{lem:final-counting-algorithm}]
    From \cref{lem:partition-of-set-of-bkc-AMOs}, $\#\text{AMO}(G,\mathcal{K}) = \sum_{C\in \Pi(G)}{|\{\alpha: \alpha \in \text{AMO}(G,\mathcal{K}) \text{ and } C=C_{\alpha}\}|}$. In \cref{sec:reduction-of-the-problem}, after defining $\mathcal{K}$-consistency of $G^C$ (\cref{def:bc-CGC}), we show that for any maximal clique $C$ of $G$, if $G^C$ is not $\mathcal{K}$-consistent then ${|\{\alpha: \alpha \in \text{AMO}(G,\mathcal{K}) \text{ and } C=C_{\alpha}\}|} = 0$. From \cref{lem:formular-for-KT-consistent-permutaion}, for any maximal clique $C$ of $G$, if $G^C$ is $\mathcal{K}$-consistent then for any $(\mathcal{K}, \mathcal{T})$-consistent permutation $\pi(C)$ of $C$, $|\{\alpha: \alpha \in \text{AMO}(G,\pi(C),\mathcal{K}) \text{ and }
C=C_{\alpha}\}|=\prod_{H\in \mathcal{C}_G(C)}{\#AMO(H,\mathcal{K}[H])}$. And, from \cref{lem:formular-for-non-KT-consistent-permutaion}, for any  permutation $\pi(C)$ of a maximal clique $C$ of $G$, if $\pi(C)$ is not a $(\mathcal{K}, \mathcal{T})$-consistent permutation of $C$ then $|\{\alpha: \alpha \in \text{AMO}(G,\pi(C),\mathcal{K}) \text{ and }
C=C_{\alpha}\}|=0$. This further implies for any maximal clique $C$ of $G$, if $G^C$ is $\mathcal{K}$-consistent then $|\{\alpha: \alpha \in \text{AMO}(G,\mathcal{K}) \text{ and } C=C_{\alpha}\}| = \Phi(C,FP(C,\mathcal{T}),\mathcal{K}[C]) \times
  {\prod_{H\in \mathcal{C}_G(C)}{\#AMO(H,\mathcal{K}[H])}}$, where  $\Phi(C,FP(C,\mathcal{T}),\mathcal{K}[C])$ is the number of $(\mathcal{K}, \mathcal{T})$-consistent permutations of a maximal clique $C$ of $G$ (from \cref{def:background-consistent-pi-C,def:phi-S-R-K}). All these things further imply 
$\#\text{AMO}(G,\mathcal{K}) = \sum_{C: G^C \text{ is } \mathcal{K}\text{-consistent}}
  {\Phi(C,FP(C,\mathcal{T}),\mathcal{K}[C]) \times
  {\prod_{H\in \mathcal{C}_G(C)}{\#AMO(H,\mathcal{K}[H])}}}$.
This proves the correctness of \cref{lem:final-counting-algorithm}.
\end{proof}

\begin{proof}[Proof of \cref{lem:phi-computaion}]
  Proof of \cref{item-3-of-phi-computaion}: If $R=\emptyset$, then a permutation
  $\pi(S)$ of $S$ is $\mathcal{K}$-consistent if ordering of $V_{\mathcal{K}}$
  in $\pi(S)$ is $\mathcal{K}$-consistent.
  $\Psi(V_{\mathcal{K}}, \mathcal{K})$ gives the number of
  $\mathcal{K}$-consistent permutations of $V_{\mathcal{K}}$. Number of
  permutations of $S$ that has the same ordering of $V_{\mathcal{K}}$ in it is
  $\frac{|S|!}{|V_{\mathcal{K}}|!}$. This completes the proof of
  \cref{item-3-of-phi-computaion}.

  Proof of \cref{item-1-of-phi-computaion}: If there exists an edge
  $u \rightarrow v\in \mathcal{K}$ such that $u\in S\setminus R_l$ and $v\in R_l$ then no
  $\mathcal{K}$-consistent permutation of $S$ exists that starts with $R_l$.

  Proof of \cref{item-2-of-phi-computaion}: If there does not exist an edge
  $u \rightarrow v\in \mathcal{K}$ such that $u\in S\setminus R_l$ and $v\in R_l$, then
  one way to compute $\Phi(S,R,\mathcal{K})$ is to first compute number of
  $\mathcal{K}$-consistent permutations of $S$ that do not start with
  $R_1,R_2,\ldots, R_{l-1}$, i.e.,
  $\Phi(S,R-\{R_l\},\mathcal{K})$. But, $\Phi(S,R-\{R_l\},\mathcal{K})$ also counts the $\mathcal{K}$-consistent permutations of $S$ that starts with $R_l$
  but not with any $R_i$, for $1\leq i<l$. We subtract such permutations from
  $\Phi(S,R-R_l,\mathcal{K})$. To construct a $\mathcal{K}$-consistent
  permutation of $S$ that starts with $R_l$ but does not start with any $R_i$, $1\leq i<l$, we have to first construct a permutation of $R_l$ that does
  not start with any $R_i$, $1\leq i<l$, and then we have to construct a
  $\mathcal{K}$-consistent permutation of the remaining vertices of $S$. This
  implies the number of $\mathcal{K}$-consistent permutations of $S$ that start
  with $R_l$ and not with any $R_i$, $1\leq i<l$, is
  $\Phi(R_l, R-\{R_l\},\mathcal{K})\times \Phi(S\setminus
  R_l,\emptyset,\mathcal{K})$.
\end{proof}

\begin{proof}[Proof of \cref{lem:alg-valid-perm-correctness}]
    If $R=\emptyset$ then from \cref{item-3-of-phi-computaion} of \cref{lem:phi-computaion}, $\Phi(S,R,\mathcal{K})= \frac{|S|!}{|V_{\mathcal{K}}|!} \times \Psi(V_{\mathcal{K}}, \mathcal{K})$. Line \ref{alg-vp:return-K-empty} of \cref{alg:valid-permutaion-of-clique} returns the same.

    If $R=\{R_1,R_2,\ldots, R_l\} \neq \emptyset$, and there exist an edge $u\rightarrow v \in \mathcal{K}$ such that $u\in S\setminus R_l$ and $v\in R_l$, then from \cref{item-1-of-phi-computaion} of \cref{lem:phi-computaion}, $\Phi(S,R,\mathcal{K})=\Phi(S,R-R_l,\mathcal{K})$. Line \ref{alg-vp:return2} of \cref{alg:valid-permutaion-of-clique} returns the same.

    If $R=\{R_1,R_2,\ldots, R_l\} \neq \emptyset$, and there does not exist an edge $u\rightarrow v \in \mathcal{K}$ such that $u\in S\setminus R_l$ and $v\in R_l$, then from \cref{item-2-of-phi-computaion} of \cref{lem:phi-computaion}, $\Phi(S,R,\mathcal{K})=\Phi(S,R-R_l,\mathcal{K})-{\Phi(R_l,R-\{R_l\},\mathcal{K}[R_l])\times \Phi(S\setminus R_l,\emptyset,\mathcal{K}[S\setminus R_l])}$. The line \ref{alg-vp:final-return} of \cref{alg:valid-permutaion-of-clique} returns the same. 
\end{proof} 

\begin{proof}[Proof of \cref{lem:main-algorithm-correct}]
  At line \ref{alg-count:clique-tree-construction}, \cref{alg:Count-AMO}
  constructs a rooted clique tree of $G$. Lines
  \ref{alg:count:G-is-a-clique}-\ref{alg-count:if-G-is-a-clique-end} deals with
  a special case when $G$ is a clique. If $R=V$ (at line
  \ref{alg:count:G-is-a-clique}) then $G$ is a clique. In this case, the number
  of $\mathcal{K}$-consistent AMOs of $G$ equals the number of
  $\mathcal{K}$-consistent permutation of $V$. Lines
  \ref{alg:count:G-is-a-clique}-\ref{alg-count:if-G-is-a-clique-end} do the
  same. For the general case (when $G$ is not a clique), we implement \cref{lem:final-counting-algorithm}. We create a queue $Q$
  at line \ref{alg-count:Queue-construction}, that stores maximal cliques of
  $G$. For each maximal clique  $C$ of $G$, we run lines \ref{alg-count:while-queue-is-not-empty}-\ref{alg-count:while-end}. 
  At line \ref{alg-count:LBFS}, we call
  \cref{alg:AMO-Union-K} (our LBFS-algorithm) for input $G,C$ and
  $\mathcal{K}$. If the first component of the output of \cref{alg:AMO-Union-K}
  is 0 then $G^C$ is not $\mathcal{K}$-consistent (from
  \cref{lem:output-of-Alg-1}). This further implies
  $|\{\alpha: \alpha\in \text{AMO}(G,\pi(C),\mathcal{K}) \text{ and }
  C=C_{\alpha}\}|=0$. This is why we skip lines
  \ref{alg-count:prod-init}-\ref{alg-count:if-end} if the first component of
  LBFS$(G,C,\mathcal{K})$ (\cref{alg:AMO-Union-K}) is 0. If the first component
  LBFS$(G,C,\mathcal{K})$ is 1 then $G^C$ is $\mathcal{K}$-consistent. In this
  case, at lines \ref{alg-count:foreach-start}-\ref{alg-count:foreach-end}, we
  compute $\prod_{H\in C_G(C)}{\#AMO(H,\mathcal{K}[H])}$, by recursively calling
  \cref{alg:Count-AMO}. At line \ref{alg-count:sum-update}, we compute
  $|\{\alpha: \alpha\in \text{AMO}(G,\mathcal{K}) \text{ and } C=C_{\alpha}\}|$
  using the discussion following \cref{lem:formular-for-KT-consistent-permutaion}. At the end of
  line-\ref{alg-count:while-end},  the variable \texttt{sum} has the value $\sum_{C : \text{$G^C$
      is $\mathcal{K}$-consistent}} {\Phi(C,FP(C,\mathcal{T}),\mathcal{K}[C])
    \times} {\prod_{H\in C_G(C)}{\#AMO(H,\mathcal{K}[H])}}$, i.e., \texttt{sum}
  equals $\#AMO(G,\mathcal{K})$ (from \cref{lem:final-counting-algorithm}). $\mathtt{memo}[G]$ stores the number of $\mathcal{K}$-consistent AMOs of $G$, once it is computed.
  This completes the proof.

\end{proof}

 \section{Proofs omitted from Section~\ref{sec:time-complexity}}
\label{appendix:time-complexity}

\begin{proof}[Proof of \cref{thm:LBFS-alg}]
  The ``While'' loop on lines \ref{alg:loop}-\ref{alg:while-end} runs at most $|V_G|$ times. Using two additional arrays (one for checking whether $v$ is in $C$ or not, and another for checking whether $v$ is in $\mathcal{L}$ or not) we can run lines \ref{alg:if-1-start}-\ref{alg:if-1-end} in $O(1)$ time. Checking the existence of edge $u\rightarrow v \in \mathcal{K}$ at line \ref{alg:mod-if} takes  $O(|\mathcal{K}|)$ time cumulatively, for all $v \in V_G$. Finding neighbors of $v$ at line \ref{alg:neighbour-v} takes $O(E_G)$ time cumulatively, for all $v\in V_G$. Partitioning each set at line \ref{alg:refine}, also takes $O(E_G)$ time cumulatively. 
Since the size of $\mathcal{K}$ can be at most $E_G$, this implies the overall time complexity of \cref{alg:AMO-Union-K} is $O(|V_G|+|E_G|)$ (with the same technique that is used to implement the standard LBFS of \cite{rose1976algorithmic}).
\end{proof}

\begin{proof}[Proof of \cref{prop:number-of-times-alg-2-called}]
\Cref{alg:Count-AMO} calls itself at line \ref{alg-count:prod-multiplication}, for each $H\in \mathcal{L}$, when $\texttt{flag} =1$. The value of \texttt{flag} is always 1 when $\mathcal{K} = \emptyset$, and $\mathcal{L} = \mathcal{C}_G(C)$ (from \cref{alg:AMO-Union-K}) does not depend on $\mathcal{K}$. This further implies that the \texttt{count} function has the maximum number of distinct recursive calls to itself when $\mathcal{K}=\emptyset$. But, for $\mathcal{K}=\emptyset$, \cref{alg:Count-AMO} is the same as the Clique-Picking algorithm of \cite{wienobst2020polynomial}, who showed that the number of these distinct recursive calls is at most $2|\Pi(G)|-1$ times. This completes the proof.
\end{proof}

\begin{proof}[Proof of \cref{lem:time-complexity-of-Phi}]
Let us denote $\mathcal{R}_i=\{R_1,R_2, \ldots , R_i\}$. For the computation of $\Phi(S, \mathcal{R}_l, \mathcal{K})$, from \cref{item-1-of-phi-computaion,item-2-of-phi-computaion} of \cref{lem:phi-computaion}, we need to compute $\Phi(S,\emptyset, \mathcal{K})$, $\Phi(S\setminus R_1,\emptyset, \mathcal{K}[S\setminus R_1])$, $\Phi(S\setminus R_2,\emptyset,\mathcal{K}[S\setminus R_2]), \ldots, \Phi(S\setminus R_l,\emptyset, \mathcal{K}[S\setminus R_l])$, $\Phi(S,\mathcal{R}_1, \mathcal{K})$, $\Phi(S,\mathcal{R}_2, \mathcal{K}), \ldots, \Phi(S,\mathcal{R}_{l-1}, \mathcal{K})$, and for each $1\leq i \leq l$,  $\Phi(R_i,\emptyset, \mathcal{K}[R_i])$, $\Phi(R_i\setminus R_1,\emptyset,\mathcal{K}[R_i\setminus R_1])$, $\ldots$ , $\Phi(R_i\setminus R_{i-1},\emptyset,\mathcal{K}[R_i\setminus R_{i-1}])$, $\Phi(R_i,\mathcal{R}_1,\mathcal{K}[R_i])$, $\Phi(R_i,\mathcal{R}_2,\mathcal{K}[R_i])$, $\ldots$ , $\Phi(R_i,\mathcal{R}_{i-1},\mathcal{K}[R_i])$.

Computation of each $\Phi(X,\phi,\mathcal{K}' = \mathcal{K}[X])$ takes $O(k!\cdot k^2)$ arithmetic operations. To see this, note that \cref{item-3-of-phi-computaion} of \cref{lem:phi-computaion} gives $\Phi(X,\phi,\mathcal{K}') = \frac{|X|!}{|V_{\mathcal{K}'}|} \times \Psi(V_{\mathcal{K}'}, \mathcal{K}')$. From our assumption, the factorials are already pre-computed. On the other hand, to compute $\Psi(V_{\mathcal{K}'},\mathcal{K}')$ we can check one-by-one the $\mathcal{K}'$-consistency of each permutation of $V_{\mathcal{K}'}$.  The number of permutations of $V_{\mathcal{K}'}$ is $O(k!)$, since $|V_{\mathcal{K}'}|\leq k$ (max-clique knowledge), while the verification of whether a permutation of $V_{\mathcal{K}'}$ is $\mathcal{K}'$-consistent or not takes $O(|\mathcal{K}'|)$ time. Since $V_{\mathcal{K}'} \leq k$, we have $|\mathcal{K}'| \leq k^2$.  Thus, the computation of $\Psi(V_{\mathcal{K}'},\mathcal{K}')$ takes $O(k^2\cdot k!)$ arithmetic operations.

Since the number of required computations of the type $\Phi(X,\phi,\mathcal{K}' = \mathcal{K}[X])$ (as already listed above) is $O(l^2)$, their total cost is $O(k!\cdot k^2 \cdot l^2)$.  After all the values listed above of the type $\Phi(X,\phi,\mathcal{K}' = \mathcal{K}[X])$ have been computed, we can implement a dynamic programming procedure using \cref{item-1-of-phi-computaion,item-2-of-phi-computaion} of \cref{lem:phi-computaion} (or, alternatively, a memoized version of \cref{alg:valid-permutaion-of-clique}) to compute all the required values of type $\Phi(X, Y\neq \emptyset, \mathcal{K}[X])$ (as listed above) using $O(1)$ arithmetic operations each.  Since the number of required computations of the type $\Phi(X, Y\neq \emptyset, \mathcal{K}[X])$ (again, as already listed above) is also $O(l^2)$, it follows that the total cost of these computations is $O(l^2)$.  Adding the computational costs for both types of computations, we see that the total cost is $O(k!\cdot k^2 \cdot l^2)$ arithmetic operations.

As the value of $l$ can be at most $|\Pi(G)|$ (the number of nodes in the clique tree $\mathcal{T}$), the overall number of arithmetic operations we need to compute $\Phi(S, \mathcal{R}_l, \mathcal{K})$ is
$O(k! \cdot k^2\cdot |\Pi(G)|^2)$.

\end{proof}

\begin{proof}[Proof of \cref{thm:ACMO-Counting-BK}]
\cref{alg:Count-AMO} is analogous to Clique-Picking Algorithm of \cite{wienobst2020polynomial}. We can also say that \cref{alg:Count-AMO} is the background knowledge version of the  Clique-Picking Algorithm of \cite{wienobst2020polynomial}. At line \ref{alg-count:clique-tree-construction}, we construct a clique tree of $G$, which takes $O(|V_G|+|E_G|)$ time.
Computing $\Phi$ function at line \ref{alg-count:memo-update}, for a clique $C$, takes $O(|\Pi(G)|^2\cdot k^2 \cdot k!)$ time (from \cref{lem:time-complexity-of-Phi}). While loop (at lines \ref{alg-count:while-queue-is-not-empty}-\ref{alg-count:while-end}) runs for $O(|\Pi(G)|)$ times, as the number of maximal cliques of $G$ is $|\Pi(G)|$. Running LBFS-algorithm (\cref{alg:AMO-Union-K}) at line \ref{alg-count:LBFS} takes $O(|V|+|E|)$ time (from \cref{thm:LBFS-alg}). Computation of function $\Phi$ at line \ref{alg-count:sum-update} takes $O({k! \cdot k^2\cdot |\Pi(G)|}^2)$ time (from \cref{lem:time-complexity-of-Phi}).

Note that it was assumed in \cref{alg:valid-permutaion-of-clique} that factorials of integers from $1$ to $\abs{V_G}$ are pre-computed.  The pre-computation of this table can be done using $O(|V_G|)$ arithmetic operations.  From \cref{prop:number-of-times-alg-2-called}, the number of distinct calls to \texttt{count} function of \cref{alg:Count-AMO} is $O(|\Pi(G)|)$.  Together, the above calculations show that the running time of \cref{alg:Count-AMO} is at most $O({k! \cdot k^2\cdot |\Pi(G)|}^4)$ or $O(k! \cdot k^2\cdot |V_G|^4)$, as for a chordal graph $G$, $|\Pi(G)|$ can be at most $|V_G|$. 
\end{proof}
 \section{Detailed explanation of Experimental Results}
\label{sec:appendix-detailed-explanation-of-experimental-results}
\textbf{Construction of chordal graphs with $n$ vertices}:  We first construct a connected \emph{Erd\H{o}s-R\'{e}nyi graph} $G$ with $n$ vertices such that each of its edges is picked with probability $p$, where $p$ is a random value in [0.1, 0.3). We give a unique rank to each node of the graph. We then process the vertices in decreasing order of rank. For each vertex $x$ of the graph, if $u$ and $v$ are two neighbors of $x$ such that $u$ and $v$ are not connected, and $u$ and $v$ both have lesser rank than the rank of $x$,  then we put an edge between $u$ and $v$ in $G$. This makes $G$ an undirected chordal graph, because, by construction, the decreasing order of ranks is a perfect elimination ordering. After this, we use rejection sampling to get an undirected connected chordal graph, i.e., if $G$ is not a connected graph then we reject $G$, and repeat the above process until we get an undirected connected chordal graph $G$.

\textbf{Construction of background knowledge edges:} For each chordal graph constructed above, and for each
  $k \in \{5,6,\ldots, 13\}$, we construct a set of background knowledge edges
  such that (i) for any maximal clique $C$ of size greater than or equal to $k$, the number of vertices of the clique that are part of an edge of the background knowledge with both endpoint in $C$ is $k$; and (ii) for any maximal clique of size less than $k$, the number of vertices of the clique that are part of an edge of background knowledge equals to the size of the clique.  To do this,  we pick one by one  each maximal clique of $G$, and select the edges of the clique such that at most $k$ vertices are involved in the set of selected edges with both of its endpoints in that clique. For this, we construct a rooted clique tree $\mathcal{T}=(T, R)$ of $G$.  We start with $R$ and then do a depth-first search (DFS) on $T$ to cover all the maximal cliques of $G$. At the iteration when we are at a maximal clique $C$, we first compute the set of vertices of the picked edges having both of their endpoints in $C$. If the size of the set is $k$ (or $|C|$, if $|C| < k$) we move to the next maximal clique. Otherwise, we one by one pick edges of $C$ and add to the selected set of edges until the set of vertices of the picked edges having both of their endpoints in $C$  reaches $k$ (or $|C|$, if $|C| < k$). 
We won't get into a situation where the set of vertices of the picked edges having both of their endpoints in $C$ exceeds $k$ (due to the clique intersection property of the clique-tree).

We write a python program for the construction of chordal graphs, background knowledge edges, and implementation of  \cref{alg:Count-AMO}.  The experiments use the open source \texttt{networkx} (\cite{SciPyProceedings_11}) package.

\paragraph{Effect of changing the size of the background knowledge while keeping
  $k$ fixed} Here we describe the construction of the background knowledge edge
sets $\mathcal{K}_1$ and $\mathcal{K}_2$ used in \cref{tab:1}.  We first
construct $\mathcal{K}_1$ as above.  We then construct $\mathcal{K}_2$ by adding
a few edges to $\mathcal{K}_1$ in such a way that the value of $k$ does not
change.  We do this experiment for different $n$ and $k$, where $n$ is the
number of nodes of the chordal graph, and $k$ is the maximum number of vertices
of any clique of the chordal graph that is part
of a background knowledge edge that lies completely inside that clique, as
defined earlier also.  \Cref{tab:2} gives a more detailed version of
\cref{tab:1} (given in the main section of the paper) with more dataset points.

\begin{table}[t]
  \centering
  \begin{tabular}{ |c|c|c|c|c|c| } 
    \hline
    $n$ & $k$ &$|\mathcal{K}_1|$ & $|\mathcal{K}_2|$ & $T_1$ & $T_2$ \\ 
    \hline
    600 & 6 & 58 & 73 & 95 & 92 \\  
    \hline
    600 & 6 & 52 & 65 & 192 & 190 \\  
    \hline
    700 & 7 & 47 & 67 & 144 & 144 \\  
    \hline
    700 & 7 & 49 & 74 & 132 & 128 \\  
    \hline
    800 & 8 & 56 & 83 & 199 & 195 \\  
    \hline
    800 & 8 & 55 & 85 & 200 & 195\\  
    \hline
    900 & 9 & 46 & 89 & 258 & 256 \\  
    \hline
    900 & 9 & 39 & 68 & 257 & 252 \\  
    \hline
    1000 & 10 & 65 & 116 & 370 & 353 \\  
    \hline
    1000 & 10 & 75 & 137 & 346 & 338 \\ 
    \hline
    1100 & 11 & 55 & 121 & 467 & 455 \\ 
    \hline
    1100 & 11 & 51 & 104 & 460 & 453 \\
    \hline
  \end{tabular}
  \caption{Exploring runtime dependence on the number of background knowledge edges: detailed table}
  \label{tab:2}
\end{table}

We can also see from the table that the running time decreases slightly by
increasing the size of background knowledge edges. This is because as the number
of background knowledge increases the number of background consistent
permutations of any clique decreases.

\end{document}